\documentclass[11pt]{article}
\usepackage[utf8]{inputenc}

\usepackage[bookmarks,colorlinks,breaklinks]{hyperref}
\hypersetup{urlcolor=blue, colorlinks=true, citecolor=green!50!black, linkcolor=blue}
\usepackage[letterpaper, left=1in, right=1in, top=0.9in, bottom=0.9in]{geometry}
\usepackage[utf8]{inputenc}
\usepackage[american]{babel}
\usepackage{etoolbox}

\usepackage[normalem]{ulem}

\usepackage{graphicx}
\graphicspath{{images/}{../images/}}

\usepackage{amsmath, amssymb, cases}
 \usepackage{theoremstyles}
\usepackage{mdframed}
\usepackage{bbm}
\usepackage{bm}

\usepackage[ruled]{algorithm}
\usepackage[noend]{algpseudocode}

\usepackage{microtype}

\usepackage{mdframed}
\usepackage{xcolor}

\usepackage{tikz}
\usetikzlibrary{fit, arrows, shapes, matrix, chains, positioning, backgrounds, arrows.meta, tikzmark, decorations.pathreplacing}

\usepackage{modletters}

\usepackage{mathtools}
\DeclarePairedDelimiter{\ceil}{\lceil}{\rceil}

\newcommand\restr[2]{{% we make the whole thing an ordinary symbol
  \left.\kern-\nulldelimiterspace % automatically resize the bar with \right
  #1 % the function
  \vphantom{\big|} % pretend it's a little taller at normal size
  \right|_{#2} % this is the delimiter
  }}

\newcommand{\eps}{\varepsilon}

\makeatletter
\def\moverlay{\mathpalette\mov@rlay}
\def\mov@rlay#1#2{\leavevmode\vtop{%
   \baselineskip\z@skip \lineskiplimit-\maxdimen
   \ialign{\hfil$\m@th#1##$\hfil\cr#2\crcr}}}
\newcommand{\charfusion}[3][\mathord]{
    #1{\ifx#1\mathop\vphantom{#2}\fi
        \mathpalette\mov@rlay{#2\cr#3}
      }
    \ifx#1\mathop\expandafter\displaylimits\fi}
\makeatother

\usepackage{graphicx}

\newcommand{\size}[1]{\mathrm{size}}

\newcommand{\set}[2][ ]{\{#2 \ifthenelse{\equal{#1}{ }}{ }{~|~#1}\}}

\newcommand{\seepage}[2][See]{
    \marginnote{
        \scriptsize {#1} p.~\pageref{#2}
    }
}

\newcommand{\reuse}[1]{
	\expandafter\stepcounter{#1_help}
    \expandafter\label{#1_app}
    \csname#1\endcsname*
}

\usepackage{mathrsfs}

\usepackage{xcolor}

\ifdefined\ShowComment
\newcommand{\danupon}[1]{{\bf \color{green} DANUPON: #1}}
\newcommand{\andres}[1]{{\bf \color{red} ANDRES: #1}}
\else
\newcommand{\andres}[1]{}
\newcommand{\danupon}[1]{}
\fi

\newcommand{\Gtrunc}{G^{\sf trunc}}

\newcommand{\cnt}{{\sf count}}
\newcommand{\fcnt}{{\sf sfcount}}

\title{Work-Optimal Parallel Minimum Cuts for Non-Sparse Graphs}
\author{Andrés López-Martínez \thanks{KTH Royal Institute of technology, Sweden, \texttt{anlm@kth.se}} \and Sagnik Mukhopadhyay \thanks{KTH Royal Institute of technology, Sweden, \texttt{sagnik@kth.se}} \and Danupon Nanongkai \thanks{KTH Royal Institute of technology, Sweden, \texttt{danupon@kth.se}}}

\date{}

\begin{document}

\begin{titlepage}
	\maketitle
	\pagenumbering{roman}
	%\vspace{-.7cm}
	\begin{abstract}

We present the first work-optimal polylogarithmic-depth parallel algorithm for the minimum cut problem on non-sparse graphs. For $m\geq n^{1+\epsilon}$ for any constant $\epsilon>0$ %\andres{for any fixed $\epsilon > \log \log n / \log n$}, 
%$\epsilon = \Omega((\log \log n) / \log n)$ 
our algorithm requires $O(m \log n)$ work and $O(\log^3 n)$ depth and succeeds with high probability. 
Its work matches the best $O(m \log n)$ runtime for sequential algorithms [MN STOC’20; GMW SOSA'21].  
This improves the previous best work by Geissmann and Gianinazzi [SPAA'18] by $O(\log^3 n)$ factor, while matching the depth of their algorithm. To do this, we design a work-efficient approximation algorithm and parallelize the recent sequential algorithms [MN STOC'21; GMW SOSA'21] that exploit a connection between 2-respecting minimum cuts and 2-dimensional orthogonal range searching.

\end{abstract}
	\newpage
	\setcounter{tocdepth}{2}
	\tableofcontents
	%\newpage
	%\listoftheorems
\end{titlepage}

\newpage
\pagenumbering{arabic}

\section{Introduction}
Computing the minimum cut, or min-cut, is a fundamental graph problem. Given a weighted undirected graph $G=(V, E)$, a cut is a set of edges whose removal disconnects $G$. The min-cut problem is to find the cut with minimum total edge weight. Throughout, we let $n=|V|$ and $m=|E|$. Unless stated otherwise, all algorithms are randomized and succeed with high probability.\footnote{We say that an algorithm succeeds with high probability (w.h.p.) if it outputs a correct answer with probability at least $1-1/n^c$ for an arbitrarily large constant $c$.}
%and $\tilde O$ hides $\poly\log(n)$. 

In the sequential setting, nearly-linear time algorithms were known since the breakthrough work of Karger \cite{Kar00}. His algorithm requires $O\left(\frac{m(\log^2 n)\log (n^2/m)}{\log\log n} + n\log^6 n\right)$ time.
This bound has been recently improved to  $O(m(\log^2 n)/\log\log n + n\log^6 n)$ \cite{gawrychowski2019minimum,mukhopadhyay2019weighted}. By simplifying the framework of \cite{mukhopadhyay2019weighted}, this was improved to %$O((m\log n + n^{1+\epsilon})/\epsilon + n\log^2 n)$ \andres{
$O(m (\log n)/ \epsilon + n^{1+2\epsilon}(\log^2 n) /\epsilon^2 + n \log^3 n)$ 
%} 
for any $\epsilon>0$ \cite{gawrychowski2020note}. For non-sparse input graphs, e.g. when $m=n^{1+\Omega(1)}$, this bound is the best in the sequential setting. (Note that better bounds exist when the input graph is very sparse or is unweighted and simple \cite{gawrychowski2019minimum, HenzingerRW17, GhaffariNT20}. These cases are relevant to our results.\footnote{For sparse input graphs, the best bound of $O(m\log^2  n)$ is due to \cite{gawrychowski2019minimum} (improving from Karger's $O(m\log^3 n)$ bound). For simple graphs, the best bounds are $O(m \log n)$ and $O(m+n \log^3 n)$ \cite{GhaffariNT20}.})
%
%DANUPON:SAY SOMETHING ABOUT SIMPLE GRAPH \sagnik{Ghaffari, Nowicki and Thorup \cite{GhaffariNT20} $O(m \log n)$ and $O(m+n \log^3 n)$ with randomized algorithms.}})
%There are also results for the special case of unweighted simple graphs. We do not consider these two cases in this paper.}

When it comes to parallel algorithms, no algorithm with nearly-linear work and polylogarithmic depth was known until the recent result by Geissmann and Gianinazzi \cite{10.1145/3210377.3210393}, where they obtain an algorithm with $O(m\log^4 n)$ work and $O(\log^3 n)$ depth. The work of this algorithm is higher than that by Karger's sequential algorithm by an $\Omega(\log n)$ factor, and it was left open in \cite{10.1145/3210377.3210393} whether  a work-optimal algorithm with polylogarithmic depth exists.

\paragraph{Our results.} We present the first work-optimal polylogarithmic-depth parallel algorithm for the minimum cut problem on non-sparse graphs. For any %$\epsilon>0$ 
$\epsilon \geq 1/ \log n$, our algorithm requires $O(\log^3 n)$ depth while its work is
%$O((m\log n + n^{1+\epsilon})/\epsilon + n\log^4 n)$ work and $O(\log^3 n)$ depth. 
%$$O\left(\frac{m \log n + n^{1+\epsilon}}{\epsilon}+\log^5 n\right)\mbox{work and $O(\log^3 n)$ depth}$$ 
%$$O\left(\frac{m \log n + n^{1+\epsilon}}{\epsilon}+n \log^5 n\right).$$ 
%
%\andres{
$$O\left(\frac{m \log n}{\epsilon} + \frac{n^{1+2\epsilon}\log^2 n}{\epsilon^2} + n \log^5 n\right)$$
%}
For non-sparse graphs ($m\geq cn\log^3 n\log \log n$ for some large constant $c$), the work of our algorithm matches the best $O(m (\log n)/ \epsilon + n^{1+2\epsilon}(\log^2 n)/\epsilon^2 + n \log^3 n)$ runtime for sequential algorithms \cite{mukhopadhyay2019weighted,gawrychowski2020note}.\footnote{This is because $m (\log n) / \epsilon + n^{1+2\epsilon}(\log^2 n) /\epsilon^2 > n \log^5 n$ when $m\geq cn\log^3 n\log \log n$. Otherwise, $m (\log n) / \epsilon\leq  n \log^5 n$ implies that $\epsilon\geq c (\log\log n)/\log(n)$. But then $n^{1+2\epsilon}(\log^2 n) /\epsilon^2\geq n^{1+2c \log\log(n)/\log(n)}\log^2 n > n\log^{5} n$ for large enough constant $c$ (note that $\epsilon<1$ since the claim is obvious otherwise). We thank Pawe{\l} Gawrychowski for the clarification regarding the bound in \cite{gawrychowski2021note}.}  
When $m=n^{1+\Omega(1)}$, the work of our algorithm can be simplified to $O(m\log n)$ and improves the previous best work by Geissmann and Gianinazzi \cite{10.1145/3210377.3210393} by an $\Omega(\log^3 n)$ factor  while matching the depth of their algorithm.

{\em Remark:} Concurrently and independently from this paper,  \cite{AndersonB2021} recently achieved a parallel algorithm with $O(\log^3 n)$ depth and $O(m\log^2 n)$ work by parallelizing the sequential algorithm of \cite{gawrychowski2019minimum}. The work of this algorithm is smaller than ours for sparse graphs ($m=O(n\log^3 n)$). This is work-optimal when $m=O(n \log^2 n)$ (but not when $m=\omega(n \log^2 n)$).\footnote{When $m=n\log^2(n)g(n)$ for some growing function $g$, we have $\frac{m \log n}{\epsilon}  + \frac{n^{1+2\eps}\log^2 n}{\eps^2}+n\log^3 n< m \log^2 n$ by setting $\epsilon$ to $c(\log\log(g(n))/\log n$ for some constant $c$.}  Table \ref{table:1} compares our result with other results.

\begin{table}[t]
\centering
\label{table:1}
\begin{tabular}{|l|c|c|}
\hline
%\textbf{Reference} & \multicolumn{2}{|c|}{\textbf{Complexity}} & \textbf{Remark}\\
\textbf{Source} & \textbf{Work} &  \textbf{Remark}\\
\hline
%\multicolumn{4}{|c|}{Sequential}\\\hline
%\multicolumn{4}{|c|}{Parallel}\\\hline
%& \textbf{Work} & \textbf{Depth} &\\\hline
\cite{10.1145/3210377.3210393} & $O(m \log^4 n)$ & Old record\\
\hline
%Here & $O(m \log^2 n + n \log^4 n)$ & $O(\log^3 n)$ & New record\\
Here & $O(m \log n + n^{1+\epsilon})$ & work-optimal on\\
(for constant $\epsilon$)& & non-sparse graphs\\
\hline
%\multicolumn{4}{|c|}{Work-optimal}\\\hline
\cite{AndersonB2021} & $O(m \log^2 n)$ & work-optimal on\\
(independent)& & sparse graphs\\

\hline
\end{tabular}
\caption{Bounds for randomized parallel algorithms computing the minimum cut with high probability. All algorithms require $O(\log^3 n)$ depth.}
\end{table}

% \begin{table}
% \caption{Bounds for randomized parallel algorithms computing the minimum cut with high probability.}% and $O(\log^3 n)$ depth.}
% \label{table:1}
% \begin{tabular}{|l|c|c|c|}
% \hline
% %\textbf{Reference} & \multicolumn{2}{|c|}{\textbf{Complexity}} & \textbf{Remark}\\
% \textbf{Source} & \textbf{Work} & \textbf{Depth} & \textbf{Remark}\\
% \hline
% %\multicolumn{4}{|c|}{Sequential}\\\hline
% %\multicolumn{4}{|c|}{Parallel}\\\hline
% %& \textbf{Work} & \textbf{Depth} &\\\hline
% \cite{10.1145/3210377.3210393} & $O(m \log^4 n)$ & $O(\log^3 n)$ & Old record\\
% %Here & $O(m \log^2 n + n \log^4 n)$ & $O(\log^3 n)$ & New record\\
% Here & $O(\frac{m \log n + n^{1+\epsilon}}{\epsilon}+\log^5 n)$ & $O(\log^3 n)$ & work-optimal on\\
% & & & non-sparse graphs\\
% %\hline
% %\multicolumn{4}{|c|}{Work-optimal}\\\hline
% \cite{AndersonB2021} & $O(m \log^2 n)$ & $O(\log^3 n)$ & work-optimal on\\
% (independent)& & & sparse graphs\\

% \hline
% \end{tabular}
% \end{table}

To achieve our result, one challenge is to first solve the problem {\em approximately}. This is crucial because all known nearly-linear time sequential algorithms require to compute the so-called {\em skeleton} and to compute the skeleton we need an $O(1)$-approximate value of the min-cut. 
While a linear-time $(2+\epsilon)$-approximation algorithm was known in the sequential setting \cite{matula1993linear}, no $O(1)$-approximation algorithm that requires polylogarithmic depth and work less than the exact algorithm of \cite{10.1145/3210377.3210393} was known in the parallel setting. 
In this paper, we show a $O(1)$-approximation algorithm with $O(m\log n+n \log^5 n)$ work and $O(\log^3 n)$ depth. The algorithm can be modified to obtain a $(1 + \eps)$-approximation of the min-cut (for any small constant $\eps$) without any change in the performance guarantee. This algorithm might be of independent interest. 

Another bottleneck in the previous parallel algorithm of \cite{10.1145/3210377.3210393} is solving the so-called {\em two-respecting cut problem}, where the randomized %\danupon{deterministic, right?} \sagnik{No.}
algorithm of \cite{10.1145/3210377.3210393} requires $O(m\log^3 n)$ work and $O(\log^2 n)$ depth. The work does not match the then-best time complexity of  $O(m\log^2 n)$ in the sequential setting \cite{Kar00}. In this paper, we obtain a work-optimal algorithm for this problem. Our algorithm is deterministic and requires %$O\left(\frac{m \log n + n^{1+\epsilon}}{\epsilon}+n\log^2 n\right)$ \andres{
$O(m / \epsilon + n^{1+2\epsilon}(\log n)/\epsilon^2 + n \log n)$ 
%} 
%\andres{The last term slightly improves the sequential algorithm. Leave it, or bound it with $n \log^2 n$?}
work and $O(\log^2 n)$ depth. Its work matches that by the sequential algorithm of \cite{mukhopadhyay2019weighted,gawrychowski2020note}. To do this, we parallelize the algorithm of \cite{mukhopadhyay2019weighted} and its simplification in \cite{gawrychowski2020note}, which exploit a connection between the 2-respecting min-cut problem and 2-dimensional orthogonal range searching.

\paragraph{Organization.} We review the necessary prerequisites in Section \ref{sec:prelims}. In Section \ref{sec:approx_mincut}, we provide the parallel algorithm for approximating min-cut in a weighted graph. Finally, in Section \ref{section.parallelMinimumCuts} we design the parallel algorithm for computing the exact minimum cut in a weighted graph.

\section{Preliminaries} \label{sec:prelims}
In this section, we introduce the model of computation and briefly state the main ideas of Karger's \cite{Kar00} and Mukhopadhyay-Nanongkai's \cite{mukhopadhyay2019weighted} min-cut algorithms. Then we review two important concepts---graph skeletons and connectivity certificates---that are extremely useful for our algorithms.

\subsection{Model of Computation}
We use the \textit{work-depth} model \cite{shiloach1982n2log, GuyEBlelloch} (sometimes called \textit{work-span} model \cite{cormen2009introduction} or \textit{work-time framework} \cite{10.5555/133889}) to design and analyze the theoretical performance of our algorithms. The \textit{work} of an algorithm is defined as the total number of operations used, similar to the time complexity in the sequential RAM model. And the \textit{depth} is the length of the longest sequence of dependent operations. We assume concurrent reads and writes are supported. 
By Brent's scheduling theorem \cite{10.1145/321812.321815}, an algorithm with work $W$ and depth $D$ takes $O(W/p + D)$ time when $p$ processors are available. A parallel algorithm is \textit{work-optimal} if it performs the same number of operations---up to a constant---as the best known sequential algorithm for the same problem.

\subsection{Karger's Minimum Cut Algorithm}
The basic idea behind Karger's randomized algorithm is to exploit a connection between minimum cuts and \textit{greedy} packings of spanning trees. 

\begin{definition}[(Greedy) tree packing]
Let $G = (V, E)$ be a weighted graph. A \textit{tree packing} $\mathcal{S}$ of $G$ is a multiset of spanning trees of $G$, where each edge $e \in E$ is loaded with the total number of trees containing $e$. We say that $\mathcal{S} = (T_1, \ldots, T_k)$ is a \textit{greedy tree packing} if each $T_i$ is a minimal spanning tree with respect to the loads induced by $\{T_1, \ldots, T_{i-1}\}$ and no edge has load greater than the weight of that edge.
\end{definition}

Such a packing $\mathcal{S}$ of $G$ has the important property that the minimum cut in $G$ 2-respects (cuts at most 2 edges of) at least a constant fraction of the trees in $\mathcal{S}$. Therefore, if we try out trees from the packing and, for each of them, find the minimum cut that 2-respects it, we will eventually stumble upon the minimum cut of $G$. To make this efficient, Karger uses random sampling techniques to construct a sparse subgraph $H$ of $G$ on which a greedy tree packing can be computed in less time. %We call graph $H$ a \textit{cut-certificate} of $G$. 
The schematic description of the algorithm is given below.

%This completes the intuition behind Karger's algorithm. %It consists of two phases: (1) a tree-packing phase, where we compute a cut-certificate of $G$ and pack appropriately many spanning trees; and (2) a cut-finding step, where we find a 2-respecting min-cut which respects a tree from the packing produced in the first phase. 

\begin{center}
  \centering
  \begin{minipage}[H]{0.8\textwidth}
\begin{algorithm}[H]
\begin{algorithmic}[1]
\State Find a sparse subgraph $H$ of $G$ that preserves the min-cut with arbitrary precision.
    %whose min-cut is small, that approximately preserves the value of the min-cut of $G$. 
\State Find a tree packing in $H$ of weight $O(\lambda') = O(\log n)$, that \textit{w.h.p.} contains a tree that 2-constrains the min-cut of $H$.
\State For each tree in the packing, find the smallest cut in $G$ that 2-respects the tree.
\end{algorithmic}
\caption[Caption for LOF]{Schematic of Karger's algorithm.} \label{algo.schematicMinCut}
\end{algorithm}
\end{minipage}
\end{center}
%From earlier introduced terminology, we identify steps 1 and 2 of the schematic as conforming to the tree packing 
We refer to steps 1 and 2 from the schematic as the \textit{tree packing step}, and to the problem of finding a 2-respecting min-cut in step 3 we refer as the \textit{cut-finding step}. 

\subsection{MN's 2-Respecting Min-Cut Algorithm} \label{section.Background.subsection.MNAlgorithm}
%We use this section to give some intuition behind 
We now briefly describe the algorithm of Mukhopadhyay and Nanongkai for serving the cut-finding step.

Let $T$ be a spanning tree of the input graph $G$. The algorithm rests on the fact that the cut determined by tree edges $e$ and $f$ is unique, and consists of edges $(u, v) \in G$ such that exactly one of $e$ and $f$ belongs to the $uv$-path in $T$. Let $cut(e,f)$ denote the value of such cut for any pair $e, f \in T$. 

Assume for now that $T$ is a path. In matrix terms, the minimum 2-respecting cut problem can be restated as finding the smallest element in the matrix $M$ of dimension $(n - 1) \times (n - 1)$ where the $(i, j)$-th entry of $M$ is determined by $M[i, j] = cut(e_i, e_j)$, with $e_i$ and $e_j$ the $i$-th and the $j$-th edges of $T$, respectively. A key insight of Mukhopadhyay and Nanongkai is in observing that the matrix $M$ satisfies a monotonicity property (see \cite[Sec. 3.1]{mukhopadhyay2019weighted} or \cite[Sec. 3.2]{gawrychowski2020note} for more details). One can take advantage of this property to find the minimum entry of $M$ in time more efficient that %more efficiently than 
inspecting up to $O(n^2)$ entries in the worst case. 
%This makes computing the minimum entry in $M$ more efficient than inspecting up to $O(n^2)$ entries in the worst case. 

This approach is generalized to handle any tree $T$ by decomposing it into a set $\mathcal{P}$ of edge-disjoint paths \cite{SLEATOR1983362}, each of which satisfies a monotonicity property (when considering all other tree edges as collapsed). Observe that it can also happen that the two tree edges $e$ and $f$ with minimum value $cut(e, f)$ each belong to a different path in $\mathcal{P}$. Hence, the algorithm must also consider pairs of paths, of which there can be $O(n^2)$ many. But another significant contribution of Mukhopadhyay and Nanongkai is in observing that one does not need to consider every pair of paths $P, Q \in \mathcal{P}$. They show that inspecting only a small subset of path pairs suffices to find the minimum 2-respecting cut of $T$ in $G$. %(See Figure \ref{fig:crossAndDownInterest} for an illustration of the notion of \textit{interest}.)

% \begin{definition}[Interest \cite{mukhopadhyay2019weighted, gawrychowski2020note}]
% Let $T_e$ denote the sub-tree of $T$ rooted at the lower endpoint (furthest from the root) of $e$, and let $\mathcal{P}$ be a collection of edge-disjoint paths in $T$.
% \begin{itemize}
%     \item An edge $e \in T$ is said to be \textit{cross-interested} in an edge $f \not\in T \setminus T_e$ if $w(T_e) < 2w(T_e, T_f)$. 
%     \item An edge $e \in T$ is said to be \textit{down-interested} in an edge $f \in T_e$ if $w(T_e) < 2w(T_f, T \setminus T_e)$.
%     \item An edge $e \in T$ is said to be interested in a path $p \in \mathcal{P}$ if it is cross-interested or down-interested in some edge of $p$. 
%     \item Two paths $p, q \in \mathcal{P}$ are said to be an \textit{interested path pair} if $p$ has an edge interested in $q$ and vice versa.
% \end{itemize}
% \end{definition}

On a very high level, the algorithm can be summarized into three steps as in the schematic below.  

\begin{center}
  \centering
  \begin{minipage}[H]{0.8\textwidth}
\begin{algorithm}[H]
\begin{algorithmic}[1]
    %\State Compute the 1-respecting min-cut of $T$ in $G$.
    \State Decompose $T$ into a set of edge-disjoint paths $\mathcal{P}$. %satisfying Property \ref{theorem:treeDecomposition}.
    \State Compute the 2-respecting min-cut for single paths in $\mathcal{P}$.
    %\State Compute a centroid decomposition of $T$.
    %\State Find interested edge-pairs and corresponding path-pairs.
    \State Compute the 2-respecting min-cut among path-pairs.
\end{algorithmic}
\caption[Caption for LOF]{Schematic of MN's algorithm.} \label{algo.schematic2RespMinCut}
\end{algorithm}
\end{minipage}
\end{center}

In Section \ref{section.parallel2Respecting} we formally argue that each step of the algorithm can be parallelized with optimal work and low depth. 

\subsection{Graph Sparsification}
In this section we review useful sparsification tools introduced by Karger \cite{Karger:Thesis} and Nagamochi and Ibaraki \cite{nagamochi1992linear} that are used in both our exact and approximation min-cut algorithms. 

\subsubsection{Graph Skeletons}
Let $G = (V, E)$ be an unweighted multigraph. A \textit{skeleton} is defined by Karger \cite{10.2307/3690490, Karger:Thesis}, as a subgraph of $G$ on the same set of vertices, obtained by placing each edge $e \in G$ in the skeleton independently with probability $p$. If $G$ is weighted, a skeleton can be constructed by placing each edge $e \in G$ in the skeleton with weight drawn from the binomial distribution with probability $p$ and the number of trials the weight $w(e)$ in $G$.

An important result %by Karger 
that rests at the core of our algorithms is the following:

\begin{theorem}[\cite{10.2307/3690490, Karger:Thesis}]
\label{theorem:samplingTheorem} 
Let $G$ be a weighted graph with minimum cut $\lambda$ and let $p = 3(d + 2)(\log n)/(\eps^2 \gamma \lambda)$ where $\eps \leq 1$, $\gamma \leq 1$ and both $\eps$ and $\gamma$ are $\Theta(1)$. Then, a skeleton $H$ of $G$ constructed with probability $p$ satisfies the following properties with high probability: 
\begin{enumerate}
    \item The minimum cut in $H$ has value $\lambda'$ within $(1 \pm \eps)$ times its expected value $p\lambda$, which is $\lambda' =  O(\log n / \eps^2)$. 
    \item The value of the minimum cut in $G$ corresponds (under the same vertex partition) to a $(1 \pm \eps)$ times\footnote{This statement hides the fact that the approximation is also scaled by the inverse of $p$. For more details see \cite[Lemma 6.3.2]{Karger:Thesis}} minimum cut of $H$.
\end{enumerate}
\end{theorem}

In simpler terms, by sampling edges of $G$ with probability $p = O(\log /\lambda)$, we obtain a graph $H$ where (i) the min-cut is $O(\log n)$, and (ii) the minimum cut of $G$ corresponds to (under the same vertex partition) $(1\pm \eps)$-times the minimum cut of $H$.

\subsubsection{Connectivity Certificates}
A related concept of importance is the \textit{sparse k-connectivity certificate}. We use this in our algorithms to bound the total size of a graph. Unlike a skeleton, where all cuts are approximately preserved, the $k$-connectivity certificate preserves cuts of values less than $k$ exactly, but cuts of a higher value are not preserved at all.

\begin{definition}
\label{definition:certificate}
Given an unweighted graph $G = (V, E)$, a \textit{sparse k-connectivity certificate} is a subgraph $H$ of $G$ with the properties:
\begin{enumerate}
    \item $H$ has at most $kn$ edges, and
    \item $H$ contains all edges crossing cuts of value $k$ or less.
\end{enumerate}
\end{definition}

Observe that the previous definition can be extended to weighted graphs if we replace an edge of weight $w$ with a set of $w$ unweighted edges with the same endpoints. Hence, the bound in size becomes a bound on the total weight of the remaining edges.  

There is one simple algorithm by Nagamochi and Ibaraki \cite{nagamochi1992linear, nagamochi1992computing} (that is enough for our purposes), which computes a sparse $k$-connectivity certificate of a graph $G$ as follows. Compute a spanning forest $F_1$ in $G$; then compute a spanning forest $F_2$ in $G - F_1$; and so on, continue computing spanning forests $F_i$ in $G - \cup_{j = 1}^{i} F_j$ 
until $F_k$ is computed. It is easy to see that the graph $H_k = \cup_{i = 1}^{k} F_i$ is $k$-connected and has $O(k n)$ edges. 

One can extend the algorithm onto the parallel setting by using e.g. the $O(m + n)$ work and $O(\log n)$ depth randomized algorithm of Halperin and Zwick \cite{HALPERIN20011} to find a spanning forest, from which the following bounds are obtained (see e.g. \cite[Sec. 2.3]{LopezMartinez1512083}). 

\begin{theorem}
\label{THEOREM:CONNECTIVITYCERTIFICATEPARALLEL}
Given an undirected weighted graph $G = (V, E)$, a $k$-connectivity certificate can be found \textit{w.h.p.} using $O(k(m + n))$ work and $O(k \log n)$ depth. 
\end{theorem}

\section{Approximating min-cut} \label{sec:approx_mincut}
We prove the following theorem:

\begin{theorem} \label{lem:approx-min-cut}
An $(1 \pm \eps)$-approximation of the minimum cut value of $G$ can be computed with work $O(m \log n + n \log^5 n)$ and depth $O(\log^3 n)$.
\end{theorem}

For simplicity of the exposition, we prove this theorem when $\eps = 1/3$. However, it is not hard to adapt the proof for any small constant $\eps$. Throughout this section, we switch between two representations of $G=(V,E)$: (i) as a weighted graph, and (ii) as an unweighted multigraph where we replace an edge $e =(u,v) \in E$ with weight $w(e)$ by $w(e)$ many unweighted copies of $e$ between the vertices $u$ and $v$. Whenever we use the term \textit{copies of an edge}, we denote the unweighted multigraph representation, and for referring to the original edges of the weighted graph $G$, we use the term \textit{weighted edges}. By \textit{sampling a weighted edge $e$ with probability $p$}, we mean sampling each of the $w(e)$ many unweighted copies of $e$ with probability $p$ independently. We use the following version of a concentration bound.

\begin{lemma}[Concentration bound] \label{lem:concentration}
Let $X_1, \cdots, X_k$ are iid boolean random variables and $X = \sum_i X_i$. Let us denote $\bbE[X] = \mu$. Then for any $\eps \in (0,1)$,
\begin{align*}
    \Pr[X \notin (1 \pm \eps) \mu] \leq 2 \cdot \exp(- \mu \cdot \eps^2/3).
\end{align*}
\end{lemma}

Consider the following family of graphs generated from $G$ by repeated sub-sampling of the unweighted copies of the edges of $G$.

\begin{definition}[Sampled hierarchy] Let $G = (V, E)$ is a weighted graph with total edge weight $W$, and let $k$ is an integer such that $2^k = W$. For any $i \in \{0, \cdots, k\}$, define $G_i$ to be a subgraph of $G_{i-1}$ where $G_i$ includes every unweighted edge of $G_{i-1}$ independently with probability $1/2$. Also, define $G_0 = G$ viewed as an unweighted multigraph. We denote the family $\{G_i\}_{i \in \{0,\cdots, k\}}$ to be the \textit{sampled hierarchy} of $G$.
\end{definition}

For $G$ with min-cut value $\lambda$, if we sample edges with probability $p$, then with probability $1 - o(1)$ we know that the resulting sampled graph will have min-cut value in the range $p\lambda(1 \pm \varepsilon)$ (given that $p \lambda = \Omega(\log n)$). We define the \textit{skeleton sampling probability} in the following way:

\begin{definition}
For a graph $G = (V, E)$ with min-cut value $\lambda$, define the skeleton sampling probability $p_s = 100 \log n/\lambda$.
\end{definition}

\begin{definition}[Skeleton layer]
Given a sampled hierarchy, we denote the skeleton layer as layer $s$ such that $2^{-s} = p_s$.
\end{definition}

Note that, at this point, we do not know that value of $\lambda$ of $G$, but we do know the following facts which follow from standard concentration argument (Lemma \ref{lem:concentration}, setting $\eps = 1/4$ for Claim \ref{clm:samp-hier-skeleton} and $\eps = 1/5$ and $1/3$ respectively for Claim \ref{clm:samp-hier-non-skeleton}).

\begin{claim} \label{clm:samp-hier-skeleton}
If we sample the edges of $G$ with probability $p_s$, then the value of the min-cut in $G_s$ is between $[75 \log n, 125 \log n]$ with high probability.
\end{claim}

\begin{claim}\label{clm:samp-hier-non-skeleton}
If we sample edges of $G$ with probability at least $2p_s$, then the value of the min-cut in the sampled graph is at least $160 \log n$. Similarly, if we sample with probability at most $p_s/2$, then the value of the min-cut in the sampled graph is at most $67 \log n$.
\end{claim}

A word of caution regarding the proof of Claim \ref{clm:samp-hier-non-skeleton}: For really small sampling probability, we may need to use the additive form of concentration bound which says the following: For any $\eps > 0$,
\begin{align*}
    \Pr[X > \mu + \eps] \leq \exp(- 2\eps^2/k),
\end{align*}
where we use the same notation as that in Lemma \ref{lem:concentration}. Hence, if we can compute the min-cut value in every $G_i$ in the sampled hierarchy, we can find out a $(1 \pm 1/4)$ approximation of the skeleton probability $p_s$ and, thereby, a $(1 \pm 1/3)$-approximation of $\lambda$. In the rest of this section, we elaborate on how to find these min-cut values.

\subsection{Computing min-cut in sampled hierarchy}

Of course, computing minimum cut na\"ively in each $G_i$ is not work efficient. Towards designing an efficient algorithm, we first define the \textit{critical layer} in the hierarchy for every weighted edge $e$.

\begin{definition}[Critical layer]
For a weighted edge $e \in E(G)$ with weight $w(e)$, define the critical layer $t_e$ in the sampled hierarchy w.r.t. the edge $e$ as the largest integer such that
\begin{align*}
\frac{w(e)}{2^{t_e}} \geq 500 \log n.
\end{align*}
\end{definition}

Given the notion of the \textit{critical layer}, we modify the sampled hierarchy in the following way:

\begin{definition}[Truncated hierarchy]
Given a graph $G = (V, E)$ and its corresponding sampled hierarchy $\{G_i\}_i$, we obtain the truncated hierarchy by removing all unweighted copies of any edge $e$ other than what is already present in $G_{t_e}$ from the layers $G_i, i < t_e$. We denote this hierarchy as $\{\Gtrunc_i\}_i$.
\end{definition}

The intuition behind \textit{truncating} the sampled hierarchy in this way is the following: In the graph $G_s$ of the sampled hierarchy (where $s$ is the skeleton layer), we know that the min-cut value is at most $125 \log n$ (Claim \ref{clm:samp-hier-skeleton}) and, hence, the number of unweighted copies of any weighted edge $e$ participating in the min-cut of $G_s$ is at most $125 \log n$. Hence it should not be a problem to remove any extra unweighted copies of $e$ from the hierarchy as they are useless as long as the min-cut of $G_s$ is concerned. We make this intuition concrete in the next section.

\subsubsection{Properties of the truncated hierarchy}

First, we bound the number of unweighted copies of any edge $e$ available in the truncated hierarchy. Note that, because of our definition of the critical layer, the expected number of copies of an edge $e$ in layer $t_e$ is $500 \log n$. The following claim follows from standard concentration argument (Lemma \ref{lem:concentration}, by setting $\eps = 1/5$).

\begin{claim} \label{clm:copies-in-trun-hier}
The number of unweighted copies of any edge $e$ in the critical layer $t_e$ of the truncated hierarchy is between $[400 \log n, 600 \log n]$ with high probability.
\end{claim}

Note that we need a guarantee that the min-cut in $\Gtrunc_s$ is \textit{well separated} from the min-cut values of $\Gtrunc_i$'s above and below the skeleton layer in order for us to find out where the skeleton layer is. The next three claims give us this guarantee.

\begin{claim}
For the skeleton layer $s$, the value of the min-cut in $\Gtrunc_s$ is $[75 \log n, 125 \log n]$ with high probability.
\end{claim}

\begin{proof}
This follows from the facts that (i) in the sampled hierarchy $\{G_i\}_i$, the value of the min-cut in $G_s$ is in the same range (Claim \ref{clm:samp-hier-skeleton}), and (ii) for every edge $e$ taking part in the min-cut, $t_e < s$. This is because, for any such edge $e$, the number of unweighted copies of $e$ in the min-cut of $G_s$ is at most $125 \log n$ w.h.p. whereas, in the critical layer $t_s$, the number of such copies is at least $400 \log n$ w.h.p.
\end{proof}

A similar argument also proves the following claim.
\begin{claim}
For layers $i > s$, the value of the min-cut in $\Gtrunc_s$ is at most $67 \log n$.
\end{claim}

For layers numbered less than $s$, we have to be careful because we may have removed edges in the truncated hierarchy. Nevertheless, we can prove the following claim which is enough for our purpose.

\begin{claim}
For layer $i < s$, the value of the min-cut in $\Gtrunc_i$ is at least $160 \log n$.
\end{claim}

\begin{proof}
We can do an exactly similar argument for the layer $s-1$ to conclude that the min-cut value in $\Gtrunc_{s-1}$ is at least $160 \log n$. The claim follows from the fact that the min-cut value cannot decrease as $\Gtrunc_{i} \subseteq \Gtrunc_{i-1}$ and the min-cut is a monotone function.
\end{proof}

\subsubsection{Obtaining the truncated hierarchy}

\begin{center}
  \centering
  \begin{minipage}[H]{0.8\textwidth}
\begin{algorithm}[H]
\caption{\textsc{Truncated (\& Exclusive) hierarchy computation}}\label{alg:hier-trunc}
\begin{algorithmic}[1]
\For{every edge $e$ in $G$}
    \State Compute the critical layer $t_e$.
\EndFor
\For{$i = 0$ to $k$}
    \For{all edge $e$ in $G$ with $t_e = i$}
        \State Sample binomially from $\cB(w(e), 2^{-i})$. Let the value of the random variable be $X$.
        \State Include $X$ many unweighted copies of $e$ in $\Gtrunc_i$.
    \EndFor
    \If{$i > 0$}
        \State Sample each edge of $\Gtrunc_{i-1}$ with probability $1/2$.
        \State Set $\hat G_{i-1} = \hat G_{i-1} \setminus \Gtrunc_i$.
    \EndIf
\EndFor
\end{algorithmic}
\end{algorithm}
\end{minipage}
\end{center}

\begin{claim}
The truncated hierarchy can be computed (by Algorithm \ref{alg:hier-trunc}) with at most $O(m \log n)$ work and $O(\log n)$ depth.
\end{claim}

\begin{proof}
We assume that the random variables $X$ can be sampled from their corresponding binomial distribution in $O(\log n)$ work.\footnote{This follows from \cite{KachitvichyanukulS88} where the authors show that a random variable from $\cB(p, N)$ can be sampled in $O(Np + \log N)$ work and similar depth with high probability. For our purpose, this is $O(\log n)$ as we sample only at the critical layer.} For every edge $e$, we need to sample unweighted copies binomially in layer $t_e$. This requires $O(m \log n)$ amount of work across all layers of the truncated hierarchy. In addition, the total number of edges across the hierarchy is at most $O(m \log n)$. This is because, for each edge $e$ in $G$, at most $O(\log n)$ many unweighted copies are present across the hierarchy with high probability. Hence the total work is $O(m \log n)$. The depth of the computation is bounded by the depth of the hierarchy and the depth required to sample binomial random variables, which is $O(\log n)$.
\end{proof}

\subsubsection{Certificate hierarchy}

Even after obtaining the truncated hierarchy, finding min-cut in every $\Gtrunc_i$ can turn out to be expensive. We get around this problem by constructing $O(\log n)$-cut certificate for each $\Gtrunc_i$ and finding min-cut on those certificates. This reduces the work to $O(n\  \text{poly} \log n)$ as each such certificate has only $O(n\  \text{poly} \log n)$ many edges.

\begin{definition}[Exclusive hierarchy]
Given a truncated hierarchy $\{\Gtrunc_i\}_i$, we define the exclusive hierarchy $\{\hat G_i\}_i$ as follows: \begin{itemize}
    \item $\hat G_k = \Gtrunc_k$, and
    \item $\hat G_i = \Gtrunc_i \setminus \Gtrunc_{i+1}$,
\end{itemize}
where $G \setminus H$ for two graphs $G$ and $H$ on the same set of vertices includes edges that are \textit{exclusively} present in $G$ and are not in $H$.
\end{definition}

The exclusive hierarchy can be computed while computing the truncated hierarchy at no extra cost (See Algorithm \ref{alg:hier-trunc}). For certificate construction, the main idea is to come up with certificate $H_i$ for each $\hat G_i$ and use $\bigcup_{\geq i} H_i$ as a certificate for $\Gtrunc_i$. The following algorithm does exactly that with strict budgeting on the number of times a weighted edge $e$ participates in the certificate computation.

\begin{center}
  \centering
  \begin{minipage}[H]{0.8\textwidth}
\begin{algorithm}[H]
\caption{\textsc{Certificate hierarchy computation}}\label{alg:hier-cert}
\begin{algorithmic}[1]

\For{every weighted edge $e$}
    \State Initialize $\cnt_e$ to $400 \log n$.
\EndFor
\For{$i = k$ to $0$}
    \State Initialize $\fcnt = 0$ and $H_i = \emptyset$.
    \While{$\fcnt \leq 200 \log n$ or $\hat G_i \neq \emptyset$}
        \State Remove from $\hat G_i$ all copies of any edge $e$ such that $\cnt_e = 0$.
        \State Find a spanning forest $F$ of $\hat G_i$.
        \State For every edge $e$ in $\hat G_i$, set $\cnt_e = \cnt_e -1$.
        \State Set $\hat G_i = \hat G_i\setminus F$, $H_i = H_i \cup F$.
    \EndWhile
\EndFor
\end{algorithmic}
\end{algorithm}
\end{minipage}
\end{center}

Algorithm \ref{alg:hier-cert} finds cut-certificates starting from $\hat G_k$ and moving upward in the exclusive hierarchy. In each iteration, it finds at most $200 \log n$ many spanning forests which are counted by the variable $\fcnt$. Note that, for any weighted edge $e$, the associated value $\cnt_e$ decreases in each spanning forest $F$ computation irrespective of whether any copy of $e$ is included in $F$ or not. The only way to stop the decrement of $\cnt_e$ in any iteration $i$ is to include all unweighted copies of $e$ present in $\hat G_i$ in the certificate $H_i$. This makes sure that each edge $e$ participates (and is not necessarily included) in at most $\cnt_e$ many spanning forest computations.

\begin{claim}
For any $i$, $\bigcup_{\geq i} H_i$ is a $200 \log n$-cut-certificate for $G_i$.
\end{claim}

\begin{proof}
Consider any $i$ and let $C$ be a cut of value at most $200 \log n$ in $G_i$. We need to show that this cut is maintained in $\bigcup_{\geq i} H_i$.

Consider any weighted edge $e \in C$ and its $\cnt_e$ at the $i$-th iteration. If $\cnt_e < 200 \log n$, then it means that $\bigcup_{> i}H_i$ already contains $200 \log n$ unweighted edges from $C$. This can be argued in the following way: Consider any spanning forest $F$ in $\bigcup_{> i}H_i$, computing which $\cnt_e$ decreased. Either $F$ includes an unweighted copy of $e$, or $F$ contains another unweighted edge crossing the cut $C$. In either case, $F$ contains at least 1 edge from $C$. Hence, each decrement of $\cnt_e$ corresponds to one unweighted edge crossing $C$ included in $\bigcup_{> i}H_i$.

If, on the other hand, $\cnt_e \geq 200 \log n$, then this edge is going to take part in all $200 \log n$ spanning forest computations in the $i$-th iteration. Hence $H_i$ will contain $200 \log n$ edges from $C$.
\end{proof}

We next compute the work and depth requirement for computing the certificate hierarchy.

\begin{claim}
Given the hierarchy $\hat G_i$, the certificates $\{H_i\}_i$ can be computed in work $O(m \log n)$ and depth $O(\log^3 n)$ with high probability.
\end{claim}

\begin{proof}
By the design of the algorithm, each edge takes part in at most $400 \log n$ spanning forest computation. Hence the total work is bounded by $400 \log n \times m = O(m \log n)$.

Calculating the depth is straight-forward. Note that there are $O(\log n)$ many layers in the certificate hierarchy, each amounting to $O(\log n)$ many spanning forest computation. We know that the depth required for computing a spanning forest is $o(\log n)$ \cite{HALPERIN20011}.
\end{proof}

\subsubsection{Computing min-cut on cut certificates}

\begin{claim}
The total work needed to compute min-cut on $\bigcup_{\geq i}H_i$ for all $i$ is at most $O(n \log^5 n)$ and the depth required is $O(\log^3 n)$.
\end{claim}

\begin{proof}
First note that the number of edges in $\bigcup_{\geq i}H_i$ for any $i$ is $O(n \text{poly} \log n)$. We can run either run the the algorithm designed by \cite{10.1145/3210377.3210393} or the algorithm designed in Section \ref{section.parallelMinimumCuts} to obtain the requires work and depth. For the approximation of min-cut (which is required in both of these two algorithms), we can use the expected min-cut value in $G_i$. If we use the algorithm in Section \ref{section.parallelMinimumCuts} then, to compute $O(\log n)$ many min-cuts---one for each $\bigcup_{\geq i}H_i$---the work required is $O(n \log^5 n)$ and the depth required is $O(\log^3 n)$ as we compute the $O(\log n)$ instances of min-cut parallelly. Using the algorithm of \cite{10.1145/3210377.3210393} gives worse dependence of $\text{poly} \log n$ in terms of work, but the depth remains $O(\log^3 n)$.
\end{proof}

\section{Exact Min-Cut} \label{section.parallelMinimumCuts}
%\section{Finding a Minimum 2-respecting Cut with optimal work and low depth}
In this section, we present our \textit{exact} parallel minimum cut algorithm. Following Karger's framework, the outline of the solution consists of (1) finding a tree packing and (2) for each tree in the packing, find the corresponding 2-respecting min-cut. The tree computations in part (2) of the algorithm are independent of each other, hence they can be safely executed in parallel.

Let $W_{pack}(m, n)$ and $D_{pack}(m, n)$ be the \textit{work} and \textit{depth} needed to find an appropriate packing of $O(\log n)$ trees, and let $W_{rsp}(m, n)$ and $D_{rsp}(m, n)$ be the work and depth required to find a minimum 2-respecting cut of a spanning tree. The \textit{work} and \textit{depth} bounds of a parallel min-cut algorithm can be expressed as:
\begin{align} 
W_{cut}(m, n) & = O( W_{pack}(m, n) + W_{rsp}(m, n)\log n), \label{eq:older_brother} \\
D_{cut}(m, n) & = O(D_{pack}(m, n) + D_{rsp}(m, n)). \label{eq:younger_brother}
\end{align}

In the following, we give parallel algorithms for both the tree-packing and 2-respecting steps. For the former, we consider our improved approximation algorithm of Section \ref{sec:approx_mincut} in the context of Karger's original packing procedure \cite{Kar00}. And for the latter, we rely on Mukhopadhyay and Nanongkai's minimum 2-respecting cut algorithm as described in \cite{gawrychowski2020note}. Put together via equations \eqref{eq:older_brother} and \eqref{eq:younger_brother}, our parallel algorithms imply the following overall bounds. 

\begin{theorem} \label{thm:exactMinCut}
The minimum cut in weighted graph can be computed w.h.p. using $W_{cut}(m, n) = O(m \log^2 n + n \log^5 n)$ work and $D_{cut}(m, n) = O(\log^3 n)$ depth.
\end{theorem}

%(In Section \ref{subsec:workOptimalDenseGraphs} we show how to improve this bound for non-sparse graphs.)
In the following, we first describe how to achieve this for general (weighted) graphs (Sections \ref{section.parallel2Respecting} and \ref{section.parallelPacking}). For non-sparse input graphs, in Section \ref{subsec:workOptimalDenseGraphs} we show how to further improve the work bound to $O(m \log n + n^{1+\epsilon})$. 

\subsection{Parallel 2-Respecting Min-Cut} \label{section.parallel2Respecting}
We begin with the parallelization of Mukhopadhyay and Nanongkai's \textit{simplified} 2-respecting min-cut algorithm \cite{gawrychowski2020note}. The approach follows the same structure as the sequential algorithm, except that we replace key data structures and subroutines with new parallel constructs. Together, this gives 
%a work-optimal parallel algorithm with low depth. 
an algorithm that is work-optimal with respect to its sequential counterpart and has low depth. 
%In this section we show that the 2-respecting minimum cut algorithm of Mukhopadhyay and Nanongkai can be parallelized work-optimally with low-depth. For its simplicity and slightly improved complexity, we rely on the description of the algorithm by Gawrychowski, Mozes, and Weimann \cite{gawrychowski2020note}. More precisely, we obtain the following bounds. 

\begin{theorem} \label{thm:Parallel2Respecting}
Given a spanning tree $T$ of a graph $G$, the minimum cut that 2-respects $T$ can be found \textit{w.h.p.} using $O(m \log m + n \log^3 n)$ work and $O(\log^2 n)$ depth.
\end{theorem}

We now describe how to implement each step of MN's algorithm (recall schematic from Section \ref{section.Background.subsection.MNAlgorithm}) in parallel and obtain the claimed work and depth complexities. In the following, we call the query that asks for the value $cut(e, f)$ given tree edges $e$ and $f$ a \textit{cut query}.

%The approach follows the same structure as the sequential algorithm, except that we replace key data structures and subroutines with new parallel constructs. 

\subsubsection{Tree decomposition}
The algorithm begins by partitioning the tree $T$ into a collection of edge-disjoint paths $\mathcal{P}$ with the following property: 
% such that any root-to-leaf path intersects at most $O(\log n)$ paths in $\mathcal{P}$

\begin{property}[Path Partition] \label{property:Partition}
Any root-to-leaf path in $T$ intersects $O(\log n)$ paths in $\mathcal{P}$.
\end{property}

Geissmann and Gianinazzi give a parallel algorithm \cite{10.1145/3210377.3210393} to compute such decomposition (a so-called \textit{bough decomposition}) which we can use as a black-box in our algorithm. %It requires only $O(n \log n)$ work and $O(\log^2 n)$ depth.

\begin{lemma}[{\cite[Lemma 7]{10.1145/3210377.3210393}}] \label{lemma:geissmann}
A tree with $n$ vertices can be decomposed \textit{w.h.p}\footnote{This result is originally Las Vegas, but by an application of Markov's inequality, it can easily be converted into a Monte Carlo algorithm.} into a set of edge-disjoint paths $\mathcal{P}$ satisfying Property \ref{property:Partition} using $O(n \log n)$ work and $O(\log^2 n)$ depth. 
\end{lemma}

Next, we will need a data structure for the tree decomposition which, on a node query, provides the set of paths in $\mathcal{P}$ which intersect the root-to-leaf path ending at that node. For this we give the following lemma. 
\begin{lemma} \label{lemma:1.1.2.1}
Let $\mathcal{P}$ be a set of edge-disjoint paths obtained by bough decomposition of a tree $T$. Given a tree $T$ of $n$ nodes and root $r$, there is a data structure that can preprocess $T$ with $O(n \log n)$ work and $O(\log^2 n)$ depth, and supports the following operation using $O(\log n)$ work and depth:
\begin{itemize}
    \item \textit{\textbf{Root-paths($u$)}}: given a node $u$, return an array of $O(\log n)$ disjoint paths $\mathcal{P}' \subseteq \mathcal{P}$, %(obtained by the bough decomposition of $T$), 
    such that every path $p \in \mathcal{P}'$ belongs to the same path from the root of $T$ to node $u$. %root-to-leaf path in $T$ as $e$.
\end{itemize}
\end{lemma}
\begin{proof}
We first show how to construct the data structure, and then analyze the query operation. 

\paragraph{Preprocessing.} We start by computing in parallel the postorder numbering of vertices in $T$ via the Eulerian circuit technique \cite{10.5555/133889}. 
This produces at each node $u$ the value $post(u)$ containing its rank in a postorder traversal of $T$. Next, we decompose tree $T$ in boughs, or paths, as in Lemma \ref{lemma:geissmann}. This produces an array of arrays $A$, where the element $A[i]$ consists of the sequence of edges representing a bough in $T$, arbitrarily indexed by $i$. Let $e \in A[i]$, we create the value $bough(e) = i$ to let $e$ know which bough it belongs to. This can be done trivially in parallel for every edge and every bough. The final step in the preprocessing is to sort the edges $e = (u, p(u))$ in each bough array $A[i]$ with respect to $post(p(v))$ in descending order, such that edges at shallower levels in $T$ appear first in the array. Note that $A[i][0]$ will contain the edge closest to the root in bough array $i$. %Let $\hat{e}_i = A[i][0]$. 

\medskip\noindent
\paragraph{Root-paths($u$).}
%\paragraph{Root-paths($e$).} 
Given a node $u$, we can find the desired boughs simply by walking up the tree from $u$ towards the root $r$, and keeping track of the boughs seen. Let $e = (u, p(u))$, where $p(u)$ denotes the parent of $u$. From the quantities defined during the construction, this can be done efficiently as follows:
\begin{enumerate}
    \item Initialize $i = bough(e)$.
    \item Set $\hat{e} = A[i][0]$ and append $i$ to the result list $L$. \label{step:2}
    \item If $\hat{e}$ is distinct from the root $r$, repeat step \ref{step:2} with $i = bough(p(\hat{e}))$.
    \item Return $L$.
\end{enumerate}

\medskip\noindent
\paragraph{Analysis.} 
For preprocessing, the step with the highest cost is the tree decomposition of Theorem \ref{lemma:geissmann}, which uses $O(n \log n)$ work and $O(\log^2 n)$ depth. All other steps are within the same work and depth bounds. The query algorithm performs only $O(1)$ work for each bough it finds on the way up from edge $e$ to the root of the tree. From Theorem \ref{lemma:geissmann}, there can be at most $O(\log n)$ such paths, and the result follows. 
\end{proof}

After decomposing the tree, the edges $e, f \in T$ that minimize $cut(e, f)$ can be distributed within the same path $p \in \mathcal{P}$, or in two distinct paths $p, q \in \mathcal{P}$. We now explain how to solve both cases efficiently in parallel. 

\subsubsection{Two edges in a single path} \label{subsection:SinglePath}
Let $p$ be a path in $\mathcal{P}$ of length $\ell$, and let $M_p$ be the $(\ell - 1) \times (\ell - 1)$ matrix defined by $M_p[i, j] = cut(e_i, e_j)$ with $e_i$ and $e_j$ the $i$-th and $j$-th edges of $p$. One key contribution of Mukhopadhyay and Nanongkai is in observing that the matrix $M_p$ is a \textit{Partial Monge} matrix. That is, for any $i \neq j$, it holds that $M_p[i, j] - M_p[i, j + 1] \geq M_p[i + 1, j] - M_p[i + 1, j + 1]$. To find the minimum entry in the matrix, MN give a divide-and-conquer algorithm that requires the computation of only $O(\ell \log^2 \ell)$ many cut queries. This algorithm can be shown to have a simple parallel implementation with optimal $O(\ell \log^2 \ell \cdot w_c(m))$ work and $O(d_c(m) \cdot \log \ell + \log^2 \ell)$ depth, where $w_c(m)$ and $d_c(m)$ are the work and depth required to compute a single cut query (see \cite{LopezMartinez1512083} for details). 

We can use instead an algorithm by Awarwall \textit{et al.} \cite[Thm. 2.3]{10.1145/97444.97693} that requires inspecting only $O(\ell \log \ell)$ many entries of a Partial Monge matrix using $O(\log \ell)$ depth. In Lemma \ref{lemma:dataStructureBinary} below, we show that a single cut query can be computed with optimal $w_c(m) = O(\log^2 n)$ work and $d_c(m) = O(\log n)$ depth. Thus, we can find the minimum cut determined by two edges of $p$ with $O(\ell \log^2 \ell \cdot \log^2 n)$ work and $O(\log \ell \cdot \log n)$ depth. Since paths in $\mathcal{P}$ are edge-disjoint, doing this for every path $p \in \mathcal{P}$ in parallel telescopes to the following bounds. %$O(n \log^3 n)$ total work and $O(\log^2 n)$ depth. 
%Thus, we can find the minimum 2-respecting cut restricted to single paths with total $O(n \log^3 n)$ work and $O(\log^2 n)$ depth. 

\begin{lemma} \label{lemma:parallelSinglePath}
Finding the minimum 2-respecting cut among all paths $p \in \mathcal{P}$ can be done in parallel with $O(n \log^3 n)$ work and $O(\log^2 n)$ depth.
%There is a parallel algorithm that finds the minimum 2-respecting cut among all paths $p \in \mathcal{P}$ using $O(n \log^3 n)$ work and $O(\log^2 n)$ depth. 
\end{lemma}

% \begin{lemma}[{\cite{LopezMartinez1512083}}]
% Let $G = (V, E)$ be a weighted graph, $T$ a spanning tree of $G$, and $p \in T$ a path of length $\ell$. There is a parallel algorithm that finds the minimum 2-respecting cut of $T$ in $G$ using $O(\ell \log^2 \ell \cdot w_c(m))$ work and $O(d_c(m) \cdot \log \ell + \log^2 \ell)$ depth, where $w_c(m)$ and $d_c(m)$ are the work and depth required to compute a single cut query. 
% \end{lemma}

\subsubsection{Two edges in distinct paths}
Now we look at the case when the tree edges $e, f \in T$ that minimize $cut(e, f)$ belong to different paths $p, q \in \mathcal{P}$ of combined length $\ell$. 

First observe that collapsing tree edges $e' \not\in p \cup q$ produces a residual graph with $p \cup q$ as its spanning tree, without changing the cut values determined by one edge being in $p$ and another in $q$. We could then run a similar algorithm to the previous section and find the smallest 2-respecting cut which respects one edge in $p$ and another in $q$. Doing this for every possible path-pair in $\mathcal{P}$ effectively serves to find the tree edges $e$ and $f$ that minimize $cut(e, f)$---but is not efficient, as the number of possible path-pairs is $O(n^2)$. Mukhopadhyay and Nanongkai solve this by showing that one must only inspect a small subset of \textit{interested path pairs}. (See Figure \ref{fig:crossAndDownInterest} for an illustration of the notion of \textit{interest}.)

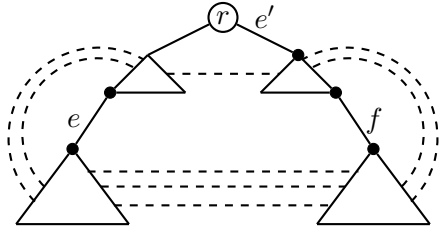
\begin{figure}
    \centering
    \begin{tikzpicture}
    
    % Tree
    \node[draw, circle, inner sep=0pt,minimum size=1em, thick] (r0) at (0, 0.25) {$r$};
    
    % Left upper tree
    \coordinate (r01) at (-1, -0.25) {}; % root
    \node[draw, fill=black,, circle, inner sep=1.5pt] (s01) at (-1.5, -0.75) {}; % left
    \coordinate (s11) at (-0.5, -0.75) {}; % right
    
    \path[draw, thick] (r01)--(s01);
    \path[draw, thick] (s01)--(s11);
    \path[draw, thick] (s11)--(r01);
    
    % Left lower tree
    \node[draw, fill=black,, circle, inner sep=1.5pt] (r02) at (-2, -1.5) {}; % root
    \coordinate (s02) at (-2.75, -2.5) {}; % left
    \coordinate (s12) at (-1.25, -2.5) {}; % right
    
    \path[draw, thick] (r02)--(s02);
    \path[draw, thick] (s02)--(s12);
    \path[draw, thick] (s12)--(r02);
    
     % Right upper tree
    \node[draw, fill=black,, circle, inner sep=1.5pt] (rr01) at (1, -0.25) {}; % root
    \coordinate (rs11) at (0.5, -0.75) {}; % left
    \node[draw, fill=black,, circle, inner sep=1.5pt] (rs01) at (1.5, -0.75) {}; % right
    
    \path[draw, thick] (rr01)--(rs01);
    \path[draw, thick] (rs01)--(rs11);
    \path[draw, thick] (rs11)--(rr01);
    
    % Right lower tree
    \node[draw, fill=black,, circle, inner sep=1.5pt] (rr02) at (2, -1.5) {}; % root
    \coordinate (rs02) at (2.75, -2.5) {}; % left
    \coordinate (rs12) at (1.25, -2.5) {}; % right
    
    \path[draw, thick] (rr02)--(rs02);
    \path[draw, thick] (rs02)--(rs12);
    \path[draw, thick] (rs12)--(rr02);
    
    % Edges
    \draw[thick] (s01) -- node[left] {$e$} (r02);
    \draw[thick] (rs01) -- node[right] {$f$} (rr02);
    \draw[thick] (r0) -- (r01);
    \draw[thick] (r0) -- node[above] {$e'$} (rr01);
    \draw[thick, dashed] (-1.625, -2.0) -- (1.625, -2.0);
    \draw[thick, dashed] (-1.45, -2.25) -- (1.45, -2.25);
    \draw[thick, dashed] (-1.8, -1.8) -- (1.8, -1.8);
    
    \draw[thick, dashed] (-0.75, -0.5) -- (0.75, -0.5);
    %\draw[] (-1.8, -1.8) -- (1.8, -1.8);
    
    % up edges left
    \draw[thick, dashed] (-2.375, -2.0) arc(-135:-300:1cm);
    \draw[thick, dashed] (-2.5, -2.20) arc(-135:-300:1.2cm);
    %\draw[] (-2.25, -1.80) arc(-135:-300:0.8cm);
    %\draw[] (-2.375, -2.0) to [bend left=90] (-1.25, -.5);
    %\draw[] (-2.55, -2.25) to [bend left=90] (-1.10, -0.35);
    
    % up edges right
    \draw[thick, dashed] (2.375, -2.0) arc(-45:120:1cm);
    \draw[thick, dashed] (2.5, -2.20) arc(-45:120:1.2cm);
    %\draw[] (2.25, -1.80) arc(-45:120:0.8cm);
    
    \end{tikzpicture}
    \caption{Example graph and spanning tree illustrating the concept of interest. Tree-edges are represented by solid edges, and non-tree edges are represented by dashed lines. The graph is unweighted and the tree is rooted at $r$. Observe that edge $e$ is cross-interested in $f$, $f$ is cross-interested in $e$, and $e'$ is down-interested in $f$.}
    \label{fig:crossAndDownInterest}
\end{figure}

\begin{definition}[Interest \cite{mukhopadhyay2019weighted, gawrychowski2020note}] \label{definition:interest}
Let $T_e$ denote the sub-tree of $T$ rooted at the lower endpoint (furthest from the root) of $e$, and let $w(T_e, T_f)$ be the total weight of edges between $T_e$ and $T_f$. In particular, we denote $w(T_e) = w(T_e, T \setminus T_e)$. Then:
\begin{enumerate}
    \item an edge $e \in T$ is said to be \textit{cross-interested} in an edge $f \not\in T \setminus T_e$ if $w(T_e) < 2w(T_e, T_f)$, 
    \item an edge $e \in T$ is said to be \textit{down-interested} in an edge $f \in T_e$ if $w(T_e) < 2w(T_f, T \setminus T_e)$,
    \item an edge $e \in T$ is said to be interested in a path $p \in \mathcal{P}$ if it is cross-interested or down-interested in some edge of $p$,
    \item two paths $p, q \in \mathcal{P}$ are said to be an \textit{interested path pair} if $p$ has an edge interested in $q$ and vice versa.
\end{enumerate}
\end{definition}

\begin{claim}[\cite{mukhopadhyay2019weighted, gawrychowski2020note}] \label{claim:pathRootToNode}
For any tree edge $e$:
\begin{enumerate}
    \item all the edges that $e$ is cross-interested in consist of a single path in $T$ going down from the root of $T$ up to some node $c_e$.
    \item all the edges that $e$ is down-interested in consist of a single path in $T$ going down from $e$ up to some node $d_e$.
\end{enumerate}
\end{claim}

Observe that, from Property \ref{property:Partition} and Claim \ref{claim:pathRootToNode}, an edge is interested in at most $O(\log n)$ paths from $\mathcal{P}$. Hence, the total number of interested path pairs is $O(n \log n)$. Finding such a subset of paths is one of the main insights of Gawrichowski, Mozes, and Wiemann's simplification \cite{gawrychowski2020note} of the original MN's algorithm. They do this through a \textit{centroid decomposition} of the tree $T$ which, for every edge $e \in T$, serves to guide the search for the nodes $c_e$ and $d_e$ that delimit the paths on which an edge can be interested in. Once these nodes have been identified, using the data structure of Lemma \ref{lemma:1.1.2.1} one may find all interested edge pairs. 

Before looking at the parallel implementation of such procedure, however,
we give the following lemma 
%in the following lemma we show how 
to efficiently determine in parallel whether an edge $e$ is interested in another edge $f$. This also serves to %show how one can 
compute the cut query $cut(e, f)$ with optimal work and low depth. 

\begin{lemma} \label{lemma:dataStructureBinary}
Given a weighted graph $G = (V, E)$ and a spanning tree $T$ of $G$, there exists a data structure that can be preprocessed in parallel with $O(m \log m)$ work and $O(\log n)$ depth, and given two edges $e$ and $f$, it can report the following with $O(\log^2 n)$ work and $O(\log n)$ depth: (1) the value $cut(e, f)$, (2) whether $e$ is cross-interested in $f$, and (3) whether $e$ is down-interested in $f$. 
\end{lemma}
\begin{proof}
Proof in the appendix.
\end{proof}

\paragraph{Finding interested path pairs}
We now look at the parallel procedure to find the subset of interested path pairs. To begin, we look at the definition of a centroid decomposition and its parallel construction. 

\begin{definition}[Centroid]
Given a tree $T$, a node $v \in T$ is a \textit{centroid} if every connected component of $T \setminus \{v\}$ consists of at most $|T|/2$ nodes. 
\end{definition}

Without loss of generality, we can assume that the input tree $T$ is a binary tree. Otherwise, simply replace a node of degree $d$ with a binary tree of size $O(d)$, where internal edges have weight $\infty$ and edges incident to leaves preserve their original weight. It is easy to see that this can be done in parallel with $O(d)$ work and $O(\log d)$ depth by constructing each tree in a bottom-up fashion. 

\begin{definition}[Centroid decomposition]
Given a tree $T$, its \textit{centroid decomposition} consists of another tree $T'$ defined recursively as follows: 
\begin{itemize}
    \item the root of $T'$ is a centroid of $T$.
    \item children of the root of $T'$ are centroids of the sub-trees that result from removing the centroid of $T$. 
\end{itemize}
\end{definition}

%Since we assumed that $T$ is binary, there can be at most $O(\log n)$ steps in the recursion. 

\begin{lemma}
The centroid decomposition of a tree $T$ of size $n$ can be computed in parallel with optimal $O(n \log n)$ work and $O(\log n)$ depth.
\end{lemma}
\begin{proof}
% To find a centroid, the classical sequential solution is to (i) select an arbitrary node $v$ of $T$, (ii) compute subtree sizes foe each node $u \in T$ and (iii) check if node $v$ is a centroid of $T$. If so, then return $v$. If it is not, then repeat (iii) on the adjacent node $w$ whose subtree $T_w$ is of greatest size; i.e., $|T_w| > |T|/2$. 

% In parallel, we can find a centroid of a tree similar to the sequential case. Steps (i) and (ii) can be done efficiently in parallel with $O(n)$ work and $O(\log n)$ depth
First, we compute sub-tree sizes for each node $u \in T$. In parallel, this can be done with optimal $O(n)$ work and $O(\log n)$ depth by the application of the \textit{Eulerian circuit} technique from parallel graph algorithms \cite{10.5555/133889}. To find a centroid, we can simply test each node for the centroid property independently in parallel. This can be done with $O(n)$ work and $O(1)$ depth, as each node reads the sub-tree sizes of at most 3 other nodes. Next, we aggregate the results in a bottom-up fashion and return any of the identified centroids. This requires $O(n)$ work and $O(\log n)$ depth. From recursing $O(\log n)$ times (in parallel) and at each level of the recursion performing only $O(n)$ work and $O(\log n)$ depth, the overall parallel complexity follows. 
\end{proof}

As mentioned previously, to find the set of interested path pairs, we identify the nodes $c_e$ and $d_e$ for every edge $e \in T$. 

\begin{claim}
For every edge $e \in T$, we can identify the nodes $c_e$ and $d_e$ in parallel using $O(n \log^3 n)$ work and $O(\log^2 n)$ depth.
\end{claim}
\begin{proof}
Let $e'$ be a tree edge, and let $c$ be the centroid node of $T$, we obtain the node $c_e$ (resp. $d_e$) as follows: (1) for the (up to) three edges $e_1$, $e_2$ and $e_3$ incident to $c$, check if edge $e'$ is cross-interested (resp. down-interested) in each of them, (2) if $e'$ is interested in $e_i$, recurse on the subtree $T_{e_i}$, (3) else, return the current centroid node as $c_{e'}$ (resp. $d_{e'}$). The correctness of the algorithm rests on %the fact that, if $e$ is interested in the edge $f = (par(c_e), c_e)$, then it is interested in all the edges from $c_e$ to the root of $T$. 
the correctness of Claim \ref{claim:pathRootToNode}.

Using Lemma \ref{lemma:dataStructureBinary} we can check with $O(\log^2 n)$ work and $O(\log n)$ depth whether edge $e$ is interested in $e_1$, $e_2$ and $e_3$. And since the depth of the recursion is $O(\log n)$ (as we assume $T$ is binary), the procedure takes $O(\log^3 n)$ work and $O(\log^2 n)$ depth for a single edge $e$. Since we may use execute this algorithm independently in parallel for each edge, the main claim follows from multiplying the work by the number of edges $n$. 
\end{proof}

\paragraph{Checking interested path pairs.}
For each pair of interested paths $p, q \in \mathcal{P}$, we want to determine the list $r = \{e_1, \ldots, e_p\}$ of edges of $p$ that are interested in $q$ and the list $s = \{f_1, \ldots, f_p\}$ of edges from $q$ that are interested in $p$. %Collapsing all tree edges $e' \not\in L_p \cup L_q$ produces a residual graph with the path defined by $L_p \cup L_q$ as its spanning tree, which we have seen gives raise to some nice properties that we can exploit. 
Observe that since the number of interested path pairs is only $O(n \log n)$, the total length of these lists is also $O(n \log n)$. 

\begin{definition}[Interest tuples]
For a tree edge $e \in p$ such that $p \in \mathcal{P}$, given its terminal nodes $c_e$ and $d_e$, we define its \textit{interest tuples} $\mathcal{T}_e$ as the list of tuples of the form $(p, q, e)$ such that $q \in \mathcal{P}$ belongs to the unique path from $c_e$ to the root of $T$, and likewise for $d_e$. 
\end{definition}

\begin{claim} \label{claim:Tuples1}
The list of interest tuples $\Gamma = \bigcup_{e \in T} \mathcal{T}_e$ can be computed in parallel using $O(n \log n)$ work and $O(\log n)$ depth.
\end{claim}
\begin{proof}
%\TODO{Complete proof.}
For an edge $e \in T$, given the nodes $c_e$ and $d_e$ we can identify the set of paths $\mathcal{P}_e$ from $\mathcal{P}$ that $e$ is cross- and down-interested in simply from issuing two \textbf{Root-paths} queries to the data structure of Lemma \ref{lemma:1.1.2.1} with nodes $c_e$ and $d_e$. This requires $O(\log n)$ work and depth per edge. From issuing all $O(n)$ queries in parallel, we get the claimed complexity bounds. Preparing the interest tuple $(p, q, e)$ on each iteration accounts for only $O(1)$ time. 
\end{proof}

\begin{lemma} \label{lemma:1.1.6}
Given a sequence $\Gamma$ of $n$ tuples of the form $(A, B, x)$, with 
$A$ and $B$ drawn from a totally ordered set $\mathcal{S}$, there exists a parallel algorithm that generates a sequence $F$ of tuples of the form $(A, \{x_1, \ldots, x_a\}, B, \{y_1, \ldots, y_b\})$ with $O(n \log n)$ work and $O(\log n)$ depth; such that for each $x_i$ the tuple $(A, B, x_i)$ appears in $\Gamma$, and for each $y_i$ the tuple $(B, A, y_i)$ appears in $\Gamma$. 
\end{lemma}
\begin{proof}
See Lemma 4.16 in \cite{LopezMartinez1512083}. 
\end{proof}

The application of Claim \ref{claim:Tuples1} together with Lemma \ref{lemma:1.1.6} serves two purposes: (1) to identify the set of interested pairs, and (2) for each such pair $p, q \in \mathcal{P}$, to identify the lists $r$ and $s$ of edges that actively participate in the interest relation. Once we have a sequence of tuples $F$ as in Lemma \ref{lemma:1.1.6}, and thus, the set of all interested path pairs, we can proceed to find the minimum 2-respecting cut as follows. 

Consider any tuple $(p, r, q, s) \in F$. Let $|r|$ and $|s|$ denote the sizes of lists $r$ and $s$, respectively, and let $\ell$ be the combined size of the two. Collapsing tree edges $e' \not\in r \cup s$ produces a residual graph with $r \cup s$ as its spanning tree, without changing the cut values determined by one edge being in $p'$ and another in $q'$. Let $M_{pq}$ be the $(|r| - 1) \times (|s| - 1)$ matrix defined by $M_{rs}[i, j] = cut(e_i, e_j)$ where $e_i$ is the $i$-th edge of $r$ and $e_j$ the $j$-th edge of $s$. Similar to single-path case, Mukhopadhyay and Nanongkai observe that the matrix $M_{rs}$ is a Monge matrix (see e.g., \cite[Claim 3.5]{mukhopadhyay2019weighted} and \cite[Lemma 2]{gawrychowski2020note} for details). That is, it satisfies that $M_{rs}[i, j] - M_{rs}[i, j + 1] \geq M_{rs}[i + 1, j] - M_{rs}[i + 1, j + 1]$ for any $i, j$ (in contrast with $i \neq j$ for single paths). 

A simplified variant of the divide-and-conquer algorithm discussed in Section \ref{subsection:SinglePath} can be easily parallelized with optimal $O(\ell \log \ell \cdot \log^2 n)$ work and $O(\log \ell \cdot \log n)$ \cite{LopezMartinez1512083}. However, we can opt for an algorithm with better work and use the parallel (and randomized) algorithm of Raman and Vishkin \cite{10.5555/314464.314661} that inspects only a linear $O(\ell)$ number of entries of a Monge matrix with $O(\log \ell)$ depth. Now observe that each tuple $f \in F$ may be processed independently in parallel. Using Lemma \ref{lemma:dataStructureBinary} for each entry inspection, and the fact that the $\sum \ell$ over all interested pairs is $O(n \log n)$, we obtain the following parallel bounds.      

\begin{lemma} \label{lemma:parallelInterestedPaths}
Finding the minimum 2-respecting cut among all interested path pairs $p, q \in \mathcal{P}$ can be done in parallel with $O(n \log^3 n)$ work and $O(\log^2 n)$ depth.
%There is a parallel algorithm that finds the minimum 2-respecting cut among all paths $p \in \mathcal{P}$ using $O(n \log^3 n)$ work and $O(\log^2 n)$ depth. 
\end{lemma}

Putting together the results of Lemmas \ref{lemma:1.1.2.1}, \ref{lemma:parallelSinglePath}, \ref{lemma:dataStructureBinary} and \ref{lemma:parallelInterestedPaths} we obtain the parallel bounds stated in Theorem \ref{thm:Parallel2Respecting}. 

%\section{Finding good trees in parallel with less work and low depth}
\subsection{Parallel Tree Packing} \label{section.parallelPacking}
Let $G = (V, E)$ be an undirected weighted graph. Karger already showed \cite[Corollary 4.2]{Kar00} that computing an appropriate tree packing in parallel can be done with $O(m \log^3 n + n \log^4 n)$ work and $O(\log^3 n)$ depth\footnote{In Karger's paper this result is stated in terms of parallel time and the number of processors that realize it.}. Using this result as a black box in Karger's tree packing framework, along with Theorem \ref{thm:Parallel2Respecting}, one can easily derive a randomized parallel min-cut algorithm that requires only $O(m \log^3 n + n \log^4 n)$ work and $O(\log^3 n)$ depth. This is already an improvement compared to Geissmann and Gianinazzi's $O(m \log^4 n)$ work algorithm when graphs are at least near-linear in size (that is, $m = \Omega(n \log n)$). However, the algorithm is not work-optimal in any setting. 

%To obtain a greedy tree packing with better work complexity and hence, a better min-cut algorithm, we follow Karger's general framework. It consists of two phases: (1) \textit{sparsifier} phase, where we compute a skeleton graph $H$ of our input graph containing fewer edges; and (2) \textit{packing} phase, where we find a set of size $O(\log n)$ (by weight) of appropriate spanning trees in $H$. 

%To achieve work optimality in non-sparse graphs, we show that a 

In this section we show that the work bound for packing an appropriate set of spanning trees can be improved in a $O(\log^2 n)$ factor. More precisely, we obtain the following: 

\begin{theorem} \label{theorem:1.2}
Given a weighted graph $G = (V, E)$, we can compute a packing $\mathcal{S}$ of $O(\log n)$ spanning trees (by weight) such that, \textit{w.h.p.} the minimum cut 2-respects at least one of the trees in $\mathcal{S}$, using $O(m \log n + n \log^5 n)$ work and $O(\log^3 n)$ depth. 
\end{theorem}

Plugging the bounds of Theorem \ref{theorem:1.2} and Theorem \ref{thm:Parallel2Respecting} into equations \eqref{eq:older_brother} and \eqref{eq:younger_brother} of Section \ref{section.parallelMinimumCuts}, we obtain the complexity claim of Theorem \ref{thm:exactMinCut}.  

%We have already stated (Section \ref{sec:prelims}) that the tree packing step of Karger's algorithm 
Recall that Karger's \cite{Kar00} tree packing step consists of two phases: (1) \textit{sparsifier} phase, where we compute a skeleton graph $H$ of our input graph containing fewer edges; and (2) \textit{packing} phase, where we find a set of size $O(\log n)$ (by weight) of appropriate spanning trees in $H$. %This latter phase is very easy to parallelize, as the only non-basic operation involved is a minimum spanning tree computation, 
The latter phase is a classical packing algorithm of Plotkin, Shmoys, and Tardos \cite{plotkin1995fast, thorup2000dynamic, neal1995randomized} and consists of a series of $O(\log^2 n)$ sequential minimum spanning tree (MST) computations. This algorithm can be easily parallelized with $O(n \log^3 n)$ work and $O(\log^3 n)$ depth from considering an appropriate MST parallel algorithm: e.g., Pettie and Ramachandran's \cite{pettie2002randomized} $O(n)$ optimal work and $O(\log n)$ depth (randomized) algorithm. This leaves the sparsifier phase as the sole bottleneck for constructing an appropriate tree packing with better work complexity. 

In the following, we assume that we have computed a constant factor underestimate $\tilde{\lambda}$ of the min-cut of $G$ via our approximation algorithm of Theorem \ref{lem:approx-min-cut}. Simply find a $(1 \pm 1/3)$-approximation $\lambda'$ of the min-cut and set $\tilde{\lambda} = \lambda' / 2$. This accounts for the additive $O(n \log^5 n)$ factor in the work bound of Theorem \ref{theorem:1.2}. %Otherwise, compute a $(1 \pm 1/3)$-approximation $\lambda'$ of the min-cut as in Theorem \ref{lem:approx-min-cut} and set $\tilde{\lambda} = \lambda' / 2$. This accounts for the additive $O(n \log^5 n)$ factor in the work bound of Theorem \ref{theorem:1.2}.

% To achieve this we follow Karger's approach consisting of two main phases: (1) \textit{sparsifier} phase, where we compute a skeleton graph $H$ of our input graph containing fewer edges, and (2) \textit{packing} phase, where we find a set of size $O(\log n)$ (by weight) of appropriate spanning trees in $H$. 

% Step 2 consists of a series of $O(\log^2 n)$ sequential minimum spanning tree (MST) computations, and can be easily parallelized with $O(n \log^3 n)$ work and $O(\log^3 n)$ depth from considering an appropriate MST parallel algorithm, e.g., Pettie and Ramachandran's \cite{pettie2002randomized} $O(n)$ optimal work and $O(\log n)$ depth (randomized) algorithm. The bottleneck is then to compute an appropriate sub-graph $H$ of $G$. 

% To achieve this, we use our approximation algorithm of Section \ref{sec:approx_mincut} to reduce the work required to construct a \textit{skeleton} graph $H$ of $G$ in parallel. Then, following Karger's framework, we can pack $O(\log n)$ many trees 

% Following Karger's approach, we can use this skeleton to find an appropriate set of spanning trees via an easy parallelization of the classical packing algorithm of Plotkin, Shmoys and Tardos \cite{plotkin1995fast, thorup2000dynamic, neal1995randomized}.

\subsubsection{Constructing Skeletons with Less Work}
Given a weighted graph $G = (V, E)$ with minimum cut $\lambda$, and an error parameter $\eps$, we are interested in finding a sparse subgraph $H = (V, E')$ on the same vertices that satisfies the following properties with high probability:

\begin{property} \label{property1}
$H$ has total weight $O(n \log n / \eps^2)$.
\end{property}
\begin{property} \label{property2}
The minimum cut in $H$ has value $\lambda' = O(\log n / \eps^2)$.
\end{property}
\begin{property} \label{property3}
the value of the minimum cut in $G$ corresponds (under the same vertex partition) to a $(1 \pm \eps)$ times minimum cut of $H$. 
\end{property}

% \begin{itemize}
%     \item \emph{Property 1:} $H$ has total weight $O(n \log n / \epsilon^2)$, 
%     \item \emph{Property 2:} the minimum cut in $H$ has value % within $(1 \pm \epsilon)$ times its expected value, of size 
%     $\lambda' = O(\log n / \epsilon^2)$, 
%     \item \emph{Property 3:} the value of the minimum cut in $G$ corresponds (under the same vertex partition) to a $(1 \pm \epsilon)$ times minimum cut of $H$. 
% \end{itemize}

In the sequential case, Karger constructs $H$ by (i) finding a sparse connectivity certificate $G'$ of $G$ with $O(n \log n)$ total weight, and then (ii) building the skeleton $H$ with bounded min-cut value $O(\log n)$. In parallel, however, this does not work as there is no parallel algorithm to construct a sparse connectivity certificate with complexity independent of the min-cut value $\lambda$. One way to remedy this parallel dependency on $\lambda$ is to first construct an appropriate skeleton of $G$ to reduce the effective minimum cut in the graph to $O(\log n)$, and then construct a sparse $k$-connectivity certificate with $k = O(\log n)$ to bound the total edge weight to $O(n \log n)$. 

%There is one challenge with this approach, however, and it is that the total weight $W$ of the input graph $G$ may be arbitrarily large. This means that sampling (naively) from the weight of $G$ takes work proportional to $W$. 
To make sampling efficient, we make the following simple observation (also appearing in \cite{bhardwaj_et_al:LIPIcs:2020:12259}).

\begin{observation} \label{observation:skeletonCap}
For a skeleton graph $H$ of $G$ satisfying Properties \ref{property2} and \ref{property3}, the weight of each edge in $H$ need not be greater than the maximum size of the minimum cut in $H$, thus $O(\log n / \eps^2)$.
\end{observation}
\begin{proof}
From Theorem \ref{theorem:samplingTheorem}, the minimum cut in $H$ will be at most $(1 + \eps)$ times its expected value $\hat{\lambda} = O(\log n / \eps^2)$ with high probability. Since edges with weight greater than $\hat{\lambda}$ never cross the minimum cut, capping their weight to a value greater than $\hat{\lambda}$ but still $O(\log n / \eps^2)$---e.g., $\ceil{c + \hat{\lambda}}$ for some constant $c \geq 2$---does not impact the minimum cut of $H$, and Properties \ref{property2} and \ref{property3} that only depend on such min-cut being approximately preserved, are still satisfied.
\end{proof}

Recall that we can build the weighted skeleton $G'$ of a graph $G$ by letting edge $e$ in $G'$ have weight drawn from the binomial distribution with probability $p$ and the number of trials weight $w(e)$ from $G$. Using \textit{inverse transform} sampling \cite{fishman1979sampling, fishman2013discrete} and Observation \ref{observation:skeletonCap}, we can avoid sampling values greater than the maximum size of the min-cut in $G'$, which is $O(\log n)$. Thus, the weight of each edge $e \in G'$ can be obtained using only $O(\log n)$ sequential work. From sampling independently in parallel for every $e \in G$, and then applying Theorem \ref{THEOREM:CONNECTIVITYCERTIFICATEPARALLEL} for the certificate construction, we obtain the following lemma.

% %Using an efficient sampling method 
% A weighted skeleton $G'$ of a graph $G$ can be built by letting edge $e$ in $G'$ have weight drawn from the binomial distribution with probability $p$ and the number of trials weight $w(e)$ from $G$. One sampling method that can be used in conjunction with Observation \ref{observation:skeletonCap} towards an efficient skeleton construction is the \textit{inverse transform} method. 

\begin{lemma}
\label{lemma:1.2.1}
Given a weighted graph $G = (V, E)$, an error parameter $\eps > 0$, and a constant underestimate of the min-cut $\tilde{\lambda}$, using $O(m \log n / \eps^2)$ work and $O(\log^2 n / \eps^2)$ depth, we can construct a sparse sub-graph graph $H = (V, E')$ satisfying properties \ref{property1}, \ref{property2} and \ref{property3}.
%satisfying the following properties with high probability:
% \begin{itemize}
%     \item \emph{Property 1:} $H$ has total weight $O(n \log n / \epsilon^2)$, 
%     \item \emph{Property 2:} the minimum cut in $H$ has value $O(\log n / \epsilon^2)$, and 
%     \item \emph{Property 3:} the value of the minimum cut in $G$ corresponds (under the same vertex partition) to a $(1 \pm \epsilon)$ times minimum cut of $H$. 
% \end{itemize}
\end{lemma}

The correctness follows from using an appropriate underestimate $\tilde{\lambda}$ of the min-cut which, by Theorem \ref{theorem:samplingTheorem}, produces a skeleton graph satisfying properties \ref{property2} and \ref{property3}. Property \ref{property1} is satisfied by the definition of sparse connectivity certificate. 

%\subsection{Improving work bound}

%\subsection{Proof of Theorem ?}

%\subsection{Proof of Theorem ?}
\subsection{Work-Optimality for Dense Graphs} \label{subsec:workOptimalDenseGraphs}
In this section we further improve the work bound of Theorem \ref{thm:exactMinCut} for non-sparse input graphs; that is, when $m = n^{1+\Omega(1)}$. %For this, we need only replace the parallel geometric data structure of Lemma \ref{lemma:dataStructureBinary} with the 2-dimensional data structure of Lemma \ref{lemma:2DimensionalEpsolinTree}. Before this, however, we design the following 1-dimensional structure.
We need the following 1- and 2-dimensional data structures:

\begin{lemma} \label{lemma:1DimensionalEpsolinTree}
For any $\epsilon > 0$, given a set $S$ of $m$ weighted points in $[n]$, there is a parallel data structure that preprocesses $S$ with $O(m/\epsilon)$ work and $O(\log n)$ depth, and reports the total weight of all points in any given interval $[x_1, x_2]$ in parallel with $O(n^\epsilon / \epsilon)$ work and $O(\log n)$ depth. 
\end{lemma}
\begin{proof}
%\TODO{Complete proof.} 
We first show how to construct the data structure, and then analyze the query operation. 
\paragraph{Preprocessing.} First we sort the points of $S$ in ascending order with $O(m)$ work and $O(\log n)$ depth using a parallel radix sort algorithm \cite{GuyEBlelloch}. Next, we construct a complete tree $C$ on $S$ with degree $n^\epsilon$, where the leaves store the points of $S$ in sorted order, and every inner node $u$ stores the tuple $(key(u), key(v))$ where $u$ and $v$ are, respectively, the leftmost and rightmost leaf nodes in the sub-tree rooted at $u$.
%The internal nodes of the tree can be constructed independently in parallel with $O(n)$ work and $O(1)$ depth.
The tree can be constructed in parallel in a bottom-up (or up-sweep) fashion starting at the leaves using $O(n)$ work and $O(1/\epsilon) < O(\log n)$ depth\footnote{This is true because the degree $n^\epsilon$ must be greater than or equal to 2 to be valid, hence $\epsilon > 1/\log n$ and $1/\epsilon \leq \log n$.}. 

%At a given level of the tree, each node $u$ can be created independently.  
Each internal node $u$ stores the value $W(u)$, consisting of the total weight of leaves in its sub-tree $C_u$. This can be computed in an up-sweep fashion with $O(n^\epsilon)$ work and $O(\epsilon \cdot \log n)$ depth. Hence, at each depth of the tree $C$, a total of $O(m)$ work and $O(\epsilon \cdot \log n)$ depth is done. And since there are a total of $O(1/\epsilon)$ levels in $C$, then the whole construction takes $O(m / \epsilon)$ work and $O(\log n)$ depth.  

\paragraph{Query.} For an internal node $u$, the idea is to inspect each of its $O(n^\epsilon)$ possible children independently in parallel and aggregate their weights in an up-sweep fashion. What we want is to identify the left $u_\ell$ and right $u_r$ children of $u$ whose leaves (the leaves on their corresponding sub-trees) are not entirely contained in the query interval $[x_1, x_2]$, and to sum the total weight contributed by each of the children in between. This can be done with $O(n^\epsilon)$ work---as node $u$ has up to $O(n^\epsilon)$ children---and $O(\epsilon \cdot \log n)$ depth. Now, each level must be processed sequentially, and the tree $C$ has $O(1/\epsilon)$ depth. Hence, the total query work is $O(n^\epsilon / \epsilon)$ and the depth is $O(\log n)$. 
\end{proof}

\begin{lemma} \label{lemma:2DimensionalEpsolinTree}
For any $\epsilon > 0$, given $m \geq n$ weighted points in the $[n] \times [n]$ grid, we can construct with work $O(m / \epsilon)$ and depth %$O(\log n / \epsilon)$ 
$O(\log^2 n)$ 
a data structure that reports the total weight of all points in any given rectangle $[x_1, x_2] \times [y_1, y_2]$ with $O(n^\epsilon/\epsilon^2)$ work and $O(\log n)$ depth. 
\end{lemma}
\begin{proof}
We show how to construct the data structure and then analyze the query operation.

\paragraph{Preprocessing.} Similar to the standard 2-dimensional \textit{range tree} data structure \cite{BENTLEY1979244}, we split our construction into two parts: (1) constructing a main, or first-level, $x$-coordinate tree $\mathscr{T}$ on $S$; and (2) at the second-level, constructing $y$-coordinate auxiliary array $A(u)$ and tree $\mathscr{T}_{aux}(u)$ on the leaves spanned by the sub-tree rooted at $u$, for each internal node $u \in \mathscr{T}$. 
\begin{enumerate}
    \item First level: We begin by sorting the points in $S$ in ascending order along their $x$-coordinates using a parallel radix sort algorithm \cite{GuyEBlelloch} with $O(m)$ work and $O(\log n)$ depth. Next, we construct the complete $O(n^\epsilon)$-degree tree $\mathscr{T}$ in parallel as in the one-dimensional case using $O(m / \epsilon)$ work and $O(\log n)$ depth. 
    \item Second level: We break this step into: (i) the construction of auxiliary arrays $A(u)$, and (ii) the construction of the auxiliary trees $\mathscr{T}_{aux}(u)$, for each node $u \in \mathscr{T}$. 
    \begin{enumerate}
        \item Auxiliary arrays: For each node $u \in \mathscr{T}$, we want an array $A_{aux}(u)$ that contains the leaf elements of the sub-tree rooted at $u$ sorted by y-coordinate. For an internal node $v$, this can be obtained from combining the sorted arrays of its children $\{A(w)|w \in children(v)\}$. One way to do this is to use radix sort on the union of elements in such arrays. This takes linear work and logarithmic depth on the total size of the arrays. Now, observe that at any given level of $\mathscr{T}$ the total size of auxiliary arrays is $m$. Hence, processing the nodes at each level of the tree independently in parallel accounts for $O(m)$ work and $O(\log n)$ depth. And since there are $O(1/\epsilon)$ levels in the tree, then the construction of the auxiliary arrays requires total $O(m / \epsilon)$ work and $O(\log n / \epsilon) < O(\log^2 n)$ depth\footnotemark[9]. 
        
        %In the binary case; that is, when $\epsilon = \log(n)$, this step can be done by merging 
        \item Auxiliary trees: For each node $u \in \mathscr{T}$, we want to build a one-dimensional structure $\mathscr{T}_{aux}(u)$ as in Lemma \ref{lemma:1DimensionalEpsolinTree} for the leaf elements of the sub-tree rooted at $u$, sorted by $y$-coordinate. With the arrays $A_{aux}(u)$ from the previous step, it is easy to construct such trees in parallel by processing the nodes at each level of the main tree $\mathscr{T}$ simultaneously. From Lemma \ref{lemma:1DimensionalEpsolinTree}, a tree $T_{aux}(u)$ can be constructed with $O(|A_{aux}|)$ work and $O(\log |A_{aux}|)$ depth. And at each level of $\mathscr{T}$, a total of $O(m)$ work is performed, hence $O(m / \epsilon)$ overall. On the other hand, the total depth telescopes to $O(\log n / \epsilon)$, which we bound as $O(\log^2 n)$.\footnotemark[9] 
    \end{enumerate}
\end{enumerate}

\paragraph{Query.} Consider a query rectangle $[x_1, x_2] \times [y_1, y_2]$. The query proceeds by searching for $x_1$ and $x_2$ in the main tree $\mathscr{T}$, and identify in the process the set $V_x$ of internal nodes $u$ whose corresponding leaves (determined by the sub-tree rooted at $u$) are entirely contained in the interval $[x_1, x_2]$ (and whose parent nodes do not have all their leaves in the interval). From this set, the 2-dimensional query is reduced to performing 1-dimensional queries (with the interval $[y_1, y_2]$) to the auxiliary structures $\mathscr{T}_{aux}$ at each node in $V_x$, and adding up the results. 
Let $S(u)$ denote the leaf elements in the sub-tree rooted at $u$. Then we are interested in finding the set $V_x = \{v \in T | S(v) \in S \cap [x_1, x_2] \text{ and } S(p(v)) \not\in S \cap [x_1, x_2]\}$, where $p(u)$ denotes the parent node of $u$. 

To find the set $V_x$ in parallel, we inspect each child of an internal node $u$ independently in parallel and test whether it is entirely contained in the interval $[x_1, x_2]$ or not. In an up-sweep fashion, we identify the "bounding" children such that all the children in between are fully contained in $[x_1, x_2]$, and we add such "bounded" children to $V_x$. This requires $O(n^\epsilon)$ work and $O(\epsilon \cdot \log n)$ depth. We then proceed the search recursively on the two "bounding" children performing the same work and depth per level of the tree, up until reaching the leaves of $\mathscr{T}$. Since there are $O(1 / \epsilon)$ levels in the tree, we have total $O(n^\epsilon / \epsilon)$ work and $O(\log n)$ depth. Observe that at each level of the tree, at most $O(n^\epsilon)$ nodes are added to the set $V_x$, hence the set has cardinality at most $O(n^\epsilon / \epsilon)$. 

Now, for each element on $V_x$ we perform a 1-dimensional query with the interval $[y_1, y_2]$ as in Lemma \ref{lemma:1DimensionalEpsolinTree} and we aggregate (sum) the results in a bottom-up fashion. Performing the $O(n^\epsilon / \epsilon)$ one-dimensional queries takes total $O(n^{2\epsilon} / \epsilon^2)$ work and $O(\log n)$ depth. And the final sum simply takes linear work and $O(\log n)$ depth, from which the claimed complexity follows. 
\end{proof}

If we replace the parallel geometric data structure of Lemma \ref{lemma:dataStructureBinary} with the 2-dimensional data structure of Lemma \ref{lemma:2DimensionalEpsolinTree} we obtain a parallel 2-respecting min-cut algorithm that has %$O(m + n \log n + n \log n \cdot n^\epsilon)$ 
$O(m / \epsilon + n^{1+2\epsilon}(\log n)/\epsilon^2 + n \log n)$ 
work and $O(\log^2 n)$ depth. Plugging these bounds into equations \eqref{eq:older_brother} and \eqref{eq:younger_brother}, along with the bounds of Theorem \ref{theorem:1.2}, results in our main theorem. 

\begin{theorem} 
% % Old (mistake)
% Given a weighted graph $G = (V, E)$ with $m = \Omega(n^{1+\epsilon})$ the minimum cut can be computed in parallel w.h.p. using $O((m \log n + n^{1+\epsilon})/\epsilon + n \log^5 n)$ work and $O(\log^3 n)$ depth, for any fixed $\epsilon > 0$. 
% New (fixed)
Given a weighted graph $G = (V, E)$ %with $m = \Omega(n^{1+\epsilon})$ 
the minimum cut can be computed in parallel w.h.p. using %$O(m \log n / \epsilon + n^{1+2\epsilon}\log^2 n /\epsilon^2 + n \log^5 n)$
$O(m (\log n)/ \epsilon + n^{1+2\epsilon}(\log^2 n)/\epsilon^2 + n \log^5 n)$ 
work and $O(\log^3 n)$ depth, for any fixed $\epsilon > 0$. 
\end{theorem}

By readjusting the parameter $\epsilon$ and suppressing constant scaling factors, we can obtain a work bound of $O(m \log n + n^{1+\epsilon})$ for any constant %$\epsilon > 0$. \andres{I think $\log \log n / \log n < \epsilon < 1$, not $\epsilon > 0$ (also true for \cite{gawrychowski2020note})} 
$\epsilon > c (\log \log n) / \log n$ for big enough $c$. If $m = n^{1+\Omega(1)}$, this can be simplified to $O(m \log n)$ work. 

% \begin{acks}
% This project has received funding from the European Research Council (ERC) under the European
% Unions Horizon 2020 research and innovation programme under grant agreement No 715672. Danupon
% Nanongkai and Sagnik Mukhopadhyay are also partially supported by the Swedish Research Council (Reg.~No. 2019-05622).
% \end{acks}

% \section*{Acknowledgment}
% This project has received funding from the European Research Council (ERC) under the European
% Unions Horizon 2020 research and innovation programme under grant agreement No 715672. Danupon
% Nanongkai and Sagnik Mukhopadhyay are also partially supported by the Swedish Research Council (Reg.~No. 2019-05622). 

%\section{Conclusion}


\bibliography{biblio}

\appendix

\section{Proof of Lemma \ref{lemma:dataStructureBinary}}
Recall that $T_e$ denotes the sub-tree of $T$ rooted at the lower endpoint (furthest from the root) of $e$, and $w(T_e, T_f)$ is the total weight of edges between $T_e$ and $T_f$. In particular, $w(T_e) = w(T_e, T \setminus T_e)$. Let $p(u)$ denote the parent of a node $u$ in $T$. Given an edge $e = (u, p(u))$ we use $u^\downarrow$ to denote the set of nodes in the tree $T_e$. 

To prove the lemma, first, we need the following data structure.

\begin{lemma} \label{dataStructRangeTreeFull}
Given an weighted graph $G = (V, E)$, and a spanning tree $T$ of $G$, there exists a data structure that can be preprocessed in parallel with $O(m \log m)$ work and $O(\log n)$ depth, and supports the following queries with $O(\log^2 n)$ work and $O(\log n)$ depth:
\begin{itemize}
    \item \textit{\textbf{cost($u$):}} given a node $u \in V$, compute $w(T_e)$, with $e = (u, p(u))$,  
    \item \textit{\textbf{cross-cost($u, v$):}} given nodes $u$ and $v$, compute $w(T_e, T_f)$, $e = (u, p(u))$ and $f = (v, p(v))$
    \item \textit{\textbf{down-cost($u, v$):}} given edges $u$ and $v$, compute $w(T_e, V \setminus T_f)$, $e = (u, p(u))$ and $f = (v, p(v))$.
\end{itemize}
\end{lemma}
\begin{proof}
We start by computing for each node $u \in V$ its postorder number in $T$, denoted by $post(u)$, in parallel with $O(n)$ work and $O(\log n)$ depth \cite{10.5555/133889}. This produces at each node $u$ its rank in a postorder traversal of $T$, denoted by $post(u)$. Additionally, we determine for each node $u \in V$ its number of descendants, denoted by $size(u)$, with $O(n)$ work and $O(\log n)$ depth \cite{10.5555/133889}.

Let $A$ be the array constructed by setting $A[post(u)] = u$ for all $u \in V$ in parallel with the same work and depth as before. From the definition of postorder traversal, we know that the postorder traversals of the children of a node $u$ appear, from left to right, immediately before $u$ in the array $A$. 
This means that the leftmost leaf $v$ of the sub-tree rooted at $u$ can be found $size(n)$ positions left of $u$ in array $A$; i.e., $post(v) = post(u) - size(u)$. We denote this position as $start(u)$. Thus:
\begin{enumerate}
    \item $\forall$ $u \in V$, $u^\downarrow$ defines the continuous range $A[start(u):post(u)]$, 
    \item $\forall$ $u \in V$, $V - u^\downarrow$ defines two continuous ranges $A[0:start(u) - 1]$ and $A[post(u) + 1:n - 1]$
\end{enumerate}

Now, assume a coordinate plane with the $x$ and $y$ axes both labeled from 0 to $n - 1$. Suppose that graph $G$ is represented as an array of edges $A_E$. We create a new array $A_P$ in parallel, such that for each $e_i = (u, v) \in A_E$ with weight $w_i$ we set $A_P[i] = (post(u), post(v))$ with the same weight $w_i$. Array $A_P$ represents points in the plane labeled by the postorder traversal of $T$. We use $A_P$ as input to construct a \textit{parallel 2-d range searching} data structure $D$ as in Lemma \ref{lemma:2DimensionalEpsolinTree}. Setting $\epsilon = \log(n)$, the preprocessing takes $O(m \log n)$ work and $O(\log n)$ depth.

From facts (1) and (2) above, at most two queries to $D$ suffice to answer each of the cut queries from the claim. Specifically,

\textit{\textbf{cost($u$):}} This query can be answered from the sum of two \textit{range addition} queries to $D$ with ranges
    \begin{equation*}
    \begin{split}
         R_1 & = [start(u), post(u)] \times [0, start(u) - 1], \\
         R_2 & = [start(u), post(u)] \times [post(u) + 1, n - 1].
    \end{split}
    \end{equation*}
\textit{\textbf{cross-cost($u, v$):}} This query can be answered with one \textit{range addition} query to $D$ with range
    \begin{equation*}
    \begin{split}
         R & = [start(v), post(v)] \times [start(u), post(u)].
    \end{split}
    \end{equation*}
\textit{\textbf{down-cost($u, v$):}} This query can be answered from the sum of two \textit{range addition} queries to $D$ with ranges 
    \begin{equation*}
    \begin{split}
         R_1 & = [start(u), post(u)] \times [0, start(v) - 1], \\
         R_2 & = [start(u), post(u)] \times [post(v) + 1, n - 1].
    \end{split}
    \end{equation*}

As for the parallel complexity, the preprocessing cost of this structure is $O(m \log m)$ work and $O(\log n)$ depth, which mostly comes from the construction of $D$. The postorder numbering and point mapping only account for additional logarithmic depth and linear work. Accessing $post$ and $start$ can be done with constant work, hence the parallel cost of each query operation simply follows from the cost of the query in Lemma \ref{lemma:2DimensionalEpsolinTree}. 
\end{proof}

\begin{lemma} \label{lemma:1.1.1.2}
The data structure from Lemma \ref{dataStructRangeTreeFull} supports weighted \textit{cut-queries} of the form $cut(e, f)$ with $O(\log^2 n)$ work and $O(\log n)$ depth. 
\end{lemma}
\begin{proof}
Given edges $e = (u', u)$ and $f = (v', v)$, let $u$ and $v$ denote the descendants of $u'$ and $v'$ respectively. There are two possibilities for $e$ and $f$: (1) they belong to different subtrees in $T$, or (2) one edge is a descendant of the other. From the $post$ and $start$ indices created in the preprocessing step of Lemma \ref{dataStructRangeTreeFull}, we can easily check which case is satisfied, since $start(v) \leq pos(u) \leq pos(v)$ iff $u$ is in the sub-tree rooted at $v$. Then, we can compute the cut query as follows:
\begin{equation*}
cut(e, f) = 
  \begin{cases} 
   w(T_e) + w(T_f) - 2w(T_e, T_f)		& \text{if } e \not\in T_f \land f \not\in T_e \\
   w(T_e) + w(T_f) - 2w(T_e, T \setminus T_f)     & \text{if } e \in T_f \\
   w(T_e) + w(T_f) - 2w(T_f, T \setminus T_u)     & \text{if } f \in T_u
\end{cases}
\end{equation*}
where queries of the form $w(T_x)$ and $w(T_x, T_y)$ can both be answered with $O(\log^2 n)$ work and $O(\log n)$ depth. Since we are doing a constant amount of such operations per cut query, the result follows. 
\end{proof}

From Lemma \ref{lemma:1.1.1.2} and the definition of interested edges (Definition \ref{definition:interest}), the claim of Lemma \ref{lemma:dataStructureBinary} follows. 

% \subsection{Appendix section}

\end{document}